\documentclass[12pt,a4paper,longbibliography,numbers,square,numbers,sort&compress]{elsarticle}

\usepackage{lineno,hyperref}
\modulolinenumbers[5]
\journal{Applied Mathematical Modelling}
\usepackage{natbib}
\usepackage{amsmath,amssymb,amsthm,amsfonts}

\usepackage{epsfig}
\usepackage{siunitx}	
\usepackage{caption}
\usepackage{subcaption}
\usepackage{float}
\usepackage{verbatim}
\usepackage[font=small]{caption}
\usepackage{soul}
\usepackage{multirow}
\usepackage{url}
\usepackage[utf8]{inputenc}
\usepackage[T1]{fontenc}
\usepackage{color}
\usepackage{threeparttable}

\newdefinition{rmk}{Remark}
\newdefinition{lemma}{Lemma}
\newdefinition{prop}{Proposition}
\usepackage{bm}
\usepackage{times}
\usepackage{mathtools}
\usepackage{gensymb} 
\newcommand{\defeq}{\vcentcolon=}
\usepackage{tablefootnote}
\usepackage{textcomp}
\usepackage{dsfont}

\usepackage{natbib}
\usepackage{wrapfig}
\usepackage{lipsum}
\usepackage{xcolor}
\usepackage[normalem]{ulem}
\usepackage{tikz, todonotes}
\usetikzlibrary{shapes,shadows,arrows}

\def\p{\partial}

\def\({\text{\huge (}}
\def\){\text{\huge )}}
\def\bf{\textbf}
\def\]{\text{\huge ]}}
\def\[{\text{\huge [}}

\def\bf{\textbf}

\newcommand{\bi}{\begin{itemize}}
\newcommand{\ei}{\end{itemize}}
\newcommand{\be}{\begin{equation}}
\newcommand{\ee}{\end{equation}}
\newcommand{\ba}{\begin{align}}
\newcommand{\ea}{\end{align}}
\newcommand{\dd}{\mathrm{d}}

\newcommand\nc{\newcommand}
\nc\pad[2]{\frac{\p #1}{\p #2}} 
\nc\padd[2]{\frac{\p^2 #1}{\p
{#2}^2}} 
\nc\nd[2]{\frac{\mathrm{d} #1}{\mathrm{d} #2}} 
\nc\ndd[2]{\frac{\mathrm{d}^2 #1}{\mathrm{d}{#2}^2}} 
\nc\pat[2]{\frac{\mathrm{D} #1}{\mathrm{D}#2}} 
\nc\ov{\bar} 
\nc\ord[1]{{\cal
O}(#1)} \nc\ra{\rightarrow} \nc\Ra{\Rightarrow} \nc\dint{{\mbox ~
d}}

\DeclareMathOperator{\Pe}{Pe}       %%% Peclet number
\DeclareMathOperator{\Da}{Da}		%%% Damkohler number
\newcommand{\bea}{\begin{eqnarray}}
\newcommand{\eea}{\end{eqnarray}}
\newcommand{\beas}{\begin{eqnarray*}}
\newcommand{\eeas}{\end{eqnarray*}}

\begin{document}

\begin{frontmatter}
\title{An analytical investigation into solute transport and sorption via intra-particle diffusion in the dual-porosity limit}

\author{Lucy C. Auton}
\address{Centre de Recerca Matem\`atica, Campus de Bellaterra, Edifici C, 08193 Bellaterra, Barcelona, Spain}
\author{Maria Aguareles\fnref{corrauth}}
\address{Department of Computer Science, Applied Mathematics and Statistics, Universitat de Girona,
Campus de Montilivi, 17071 Girona, Catalunya, Spain.}
\author{Abel Valverde}
\address{Universitat Polit\`ecnica de Catalunya, Barcelona, Spain\\ Visiting Fellow, Mathematical Institute, University of Oxford, Oxford, United Kingdom}
\author{Timothy G. Myers}
\address{Centre de Recerca Matem\`atica, Campus de Bellaterra, Edifici C, 08193 Bellaterra, Barcelona, Spain}
\author{Marc Calvo-Schwarzwalder}
\address{College of Interdisciplinary Studies, Zayed University, P.O. Box 144534 Abu Dhabi, United Arab Emirates}

\fntext[corrauth]{Corresponding author: maria.aguareles@udg.edu}
%\maketitle
\begin{abstract}
We develop a mathematical model for adsorption based on averaging the flow around, and diffusion inside, adsorbent particles in a column. The model involves three coupled partial differential equations for the contaminant concentration both in the carrier fluid and within the particle as well as the adsorption rate. The adsorption rate is modelled using the Sips equation, which is suitable for describing both physical and chemical adsorption mechanisms. Non-dimensionalisation is used to determine the controlling parameter groups as well as to determine negligible terms and so  reduce the system complexity. The inclusion of intra-particle diffusion introduces new dimensionless parameters to those found in standard works, including a form of internal Damk\"ohler number and a new characteristic time scale. We provide a numerical method for the full model and show how in certain situations a travelling wave approach can be utilized to find analytical solutions. The model is validated against available experimental data for the removal of Mercury(II)  and CO$_\text{2}$. The results show excellent agreement with measurements of column outlet contaminant concentration and provide insights into the underlying chemical reactions.
\end{abstract}

\begin{keyword}
Contaminant removal \sep Adsorption \sep Advection-diffusion \sep Mathematical Analysis \sep Asymptotics
\MSC[2010] 35Q35\sep 35B40
\end{keyword}

\end{frontmatter}

\section{Introduction}\label{sec:intro}

Column sorption is a popular practical sorption method and is used for a wide range of processes, such as the removal of emerging contaminants, volatile organic compounds, $\text{CO}_\text{2}$, dyes and salts. With this technique, pollutants are removed from a fluid by letting the mixture flow through a column filled with a porous material, the adsorbent, which captures the contaminant. The mathematical description of these processes may be traced back to the 20$^\text{th}$ century and the most commonly used model is that of Bohart and Adams \cite{Bohart20}. Recently, the validity and physical accuracy of this model have been discussed \cite{Myers24,Valverde2024,Myer22,Myer20a,Myer20b} highlighting a need for mathematical descriptions that better describe the underlying physical and chemical processes. 

Many sorption filters may be described as dual-porosity filters, \textit{i.e.}, they comprise an array of grains each of which is itself porous. This leads to two distinct regions: the `inter-particle' region (between the grains) and the 'intra-particle' region (within each porous grain). This distinction is neglected in standard models, where it is assumed that the adsorbate (\textit{i.e.}, the material being removed)  immediately attaches to the adsorbent. However, as the size of the adsorbing particles increases (for example, as the scale of the process is increased from experimental to industrial scale) it is clear that the intermediate step where the contaminant first diffuses into the adsorbent and then attaches to its inner surface has to be accounted for.  This observation is  supported by the recorded differences in qualitative behaviour of the breakthrough curves (the contaminant concentration at the column outlet) as the particle size increases \cite{Patel2022,Valverde2024}.

A classic kinetic model for adsorption is attributed to Langmuir \cite{Langmuir1918} and it is based on the assumption that the underlying  mechanism is caused by the  attraction  between a monolayer of adsorbate and the available sites on the surface of the adsorbent (physisorption). The rate of adsorption is then proportional to the contaminant concentration and  available sites. In practice, many adsorbents remove contaminants via a chemical reaction (chemisorption). If the chemical reaction is such that one molecule of contaminant reacts with one molecule of adsorbent then the Langmuir model also proves effective. However, in reality these chemical reactions can often be more complex, involving a chain of chemical reactions or a non one-to-one relation between the number of molecules of adsorbent and number of molecules of contaminant. A model proposed for these cases is the Sips model \cite{sips1948structure}, which has two  parameters that depend on the underlying chemistry (on the global order of the reaction) within the system.

In the context of physical adsorption and intra-particle diffusion, \citet{Mondal19} proposed a model that was applied to the removal of arsenic from drinking water in India, but the averaging process conducted there shows some inconsistencies that were later addressed by  \citet{Valverde2024}.  Also, \citet{Seiden19} proposed an intra-particle diffusion model for investigating the potential of microplastics to adsorb and transport contaminants. In the present paper, we propose a model to account for chemical adsorption  utilizing the Sips model, which has been shown to provide accurate results for small particles, when intra-particle diffusion is negligible \cite{Aguareles2023}. This is particularly interesting as this kinetic model reduces to more standard ones, such as the Langmuir or Freundlich models, in certain limits. To support the mass transfer model developed in \citet{Valverde2024}, we re-develop the mathematical description with a higher degree of mathematical rigor at the same time that we extend the results to adsorbent particles of arbitrary shape. 

This article is structured as follows. In \S\ref{sec:model} we develop a macroscopic model based on a series of averages of microscopic variables, and discuss the related dimensionless parameters.
In \S\ref{sec:solutionmethods} we investigate numerical and  semi-analytical solutions for physically realistic parameter combinations using a travelling wave approach. Further, we provide general results that determine in which situations such travelling wave solutions exists.
The analytical expressions are validated against experimental data in \S\ref{sec:results}. Finally, \S\ref{sec:conclusions} contains the conclusions and thoughts on possible future research.

\section{Model problem} \label{sec:model}
\begin{figure}[ht!]
    \centering
    \includegraphics[width=1\textwidth]{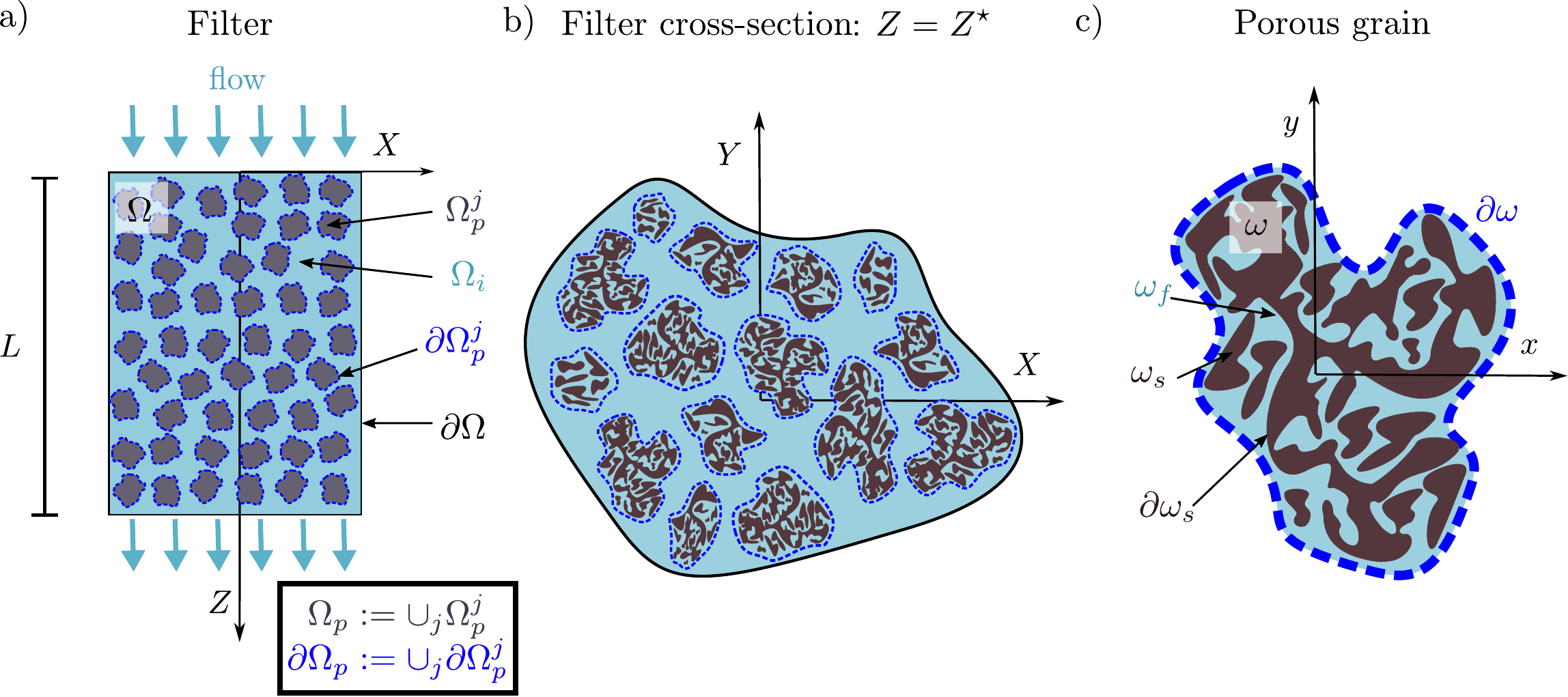}
    \caption{Flow of fluid carrying solute through a dual-porosity filter. (a)
    The fluid travels in the ${Z}$ direction with the inlet defined by ${Z}=0$ and outlet by ${Z}={L}$, (b) the cross-section of the filter is spanned by the ${X}$-${Y}$ plane. (c) Within any given identical but arbitrarily shaped three-dimensional grain, positions are determined via a local $(x,y,z)$ coordinate system.}
    \label{fig:expsetup}
\end{figure}

We consider the flow of fluid carrying a solute through a rigid porous cylinder of length ${L}$.  The solute advects, diffuses, and is removed via adsorption to the solid structure (the adsorbent).  The spatial coordinate within the cylinder is ${\bm{X}} \defeq({X},{Y},{Z})^\intercal$ with ${Z}$ the longitudinal coordinate and ${X}$ and ${Y}$ are the coordinates in the plane of the cross-section (see Figure \ref{fig:expsetup}a,b).  The fluid enters the porous medium uniformly through the inlet, at ${Z}=0$, and exits at ${Z}={L}$. 

The entire domain of the porous medium, denoted ${\Omega}$, comprises an array of $N$ identical but arbitrarily shaped porous grains. 
Within any porous grain we define a local spatial co-ordinate ${\bm{x}} \defeq({x},{y},{z})^\intercal$ ~(Figure \ref{fig:expsetup}c).  Note that although every grain is identical, when taking a cross-section the plane of intersection can contain obstacles of different sizes and shapes.

The solute molecules are assumed to be much smaller than the solid obstacles. We measure the local molar concentration of solute (amount of solute per volume of fluid \textit{surrounding the porous grains}) via the inter-particle concentration field ${c}({\bm{X}},{t})$, measured in moles/m$^3$. This is defined over the inter-particle fluid domain, denoted ${\Omega}_i$, and we impose that  ${c}({\bm{X}},{t}) \equiv 0$ inside  the porous grains.  We denote the domain comprising the union of all porous grains as ${\Omega}_p\defeq{\Omega}\setminus{\Omega}_i \equiv  \bigcup{\Omega}_p^j$, where ${\Omega}_p^j$ denotes the domain of the $j^\mathrm{th}$ porous grain (Figure \ref{fig:expsetup}a). 

As we have both local and global spatial coordinates, $\bm{x}$ and $\bm{X}$,  respectively, we also define the domain of any porous grain relative to $\bm{x}$ to be $\omega$ so that ${\bm{x}}\in {\omega} \equiv {X}\in \Omega_p^j$, for some $j$. 
Each porous grain is bounded by a fluid-fluid interface that is a smooth surface $\partial {\omega}$ that is the enclosure of the porous grain 
(see Figure \ref{fig:expsetup}c). 
We denote the local molar concentration of solute (amount of solute per volume of fluid \textit{contained within a given grain}) via the intra-particle concentration field ${c}_p^j({\bm{x}},{\bm{X}},{t})$, measured in moles/m$^3$,
which is defined within the internal fluid region denoted ${\omega}_f$. We extend the definition of ${c}^j_p$ across the entire porous particle domain by enforcing  that ${c}^j_p\equiv 0$ in ${\omega}_s\defeq {\omega}\setminus{\omega}_f$, that is the solid domain in the particle. 

We track how much solute adsorbs to the $j^\text{th}$ porous grain via a mass sink ${M}_p^j({\bm{X}},{t})$,  measured in moles,
however,  we neglect any impact of this adsorption on the size and volume of the solid, that is, we take ${\omega}_s$, and consequently every domain, to be independent of time. This is justified by the fact that the solute molecules are negligible in size relative to the obstacles and also because in this work the fluids under consideration are assumed to have low contaminant concentrations.

We define the total porosity of the cylinder to be $\Phi$, that is, the void fraction of the porous material which comprises both the void space between the grains and the void space within each grain. 
We take the average porosity within any grain to be the same  and define this value to be the `intra-particle' porosity  $\phi_p\defeq|{\omega}_f|/|{\omega}|$, for all $j$.  
Finally, we define the 
`inter-particle' porosity $\phi \defeq |{\Omega}_i|/|{\Omega}|$ to be the value that would be calculated if the grains had zero porosity (\textit{i.e.}, $ \phi\equiv\Phi \iff \phi_p=0 $). Thus we relate $\Phi$, $\phi_p$ and $\phi$ via 
\begin{equation}
    1-\Phi \equiv (1-\phi)(1-\phi_p)\, .
\end{equation}

\subsection{Full model}
We now derive equations for both the inter-particle concentration of adsorbate, $c({\bf X},t)$ and the intra-particle concentration, ${c}_p^j({\bm{x}},{\bm{X}},{t})$.

\subsubsection{Inter-particle model}
In the inter-particle region we model the flow as unidirectional and uniform along the length of the cylinder, that is, we take a constant Darcy flux ${\bm{Q}}\defeq\Phi {v}\bm{e}_Z$ such that ${v}$ is the average interstitial velocity of the fluid and where $\bm{e}_Z$ is the unit vector along the length of the filter. The transport is then described by an advection-diffusion equation 
\begin{subequations}
    \begin{equation}
    \label{itsgonemidnight}
        \frac{\partial {c}}{\partial {t}} + \frac{\partial}{\partial {Z}}({v} {c})= \bm{\nabla}_{\bm{X}}\cdot\left({\bm{\mathfrak{D}}}\cdot\bm{\nabla}_{\bm{X}}{c}\right),\quad {\bm{X}}\in{\Omega}_i\, ,
    \end{equation}
 where  
  \begin{equation}
     {\bm{\mathfrak{D}}}({\bm{X}}) \defeq 
     \begin{pmatrix}
{D}_{11}({\bm{X}}) &  {D}_{12}({\bm{X}}) & 0\\
{D}_{21}({\bm{X}}) & {D}_{22}({\bm{X}}) & 0\\
0&0&D
\end{pmatrix} = \begin{pmatrix}
    \bm{\mathfrak{D}}^{X,Y}({\bm{X}}) & 0 \\0 & D
\end{pmatrix}\, ,
 \end{equation}
 is the effective diffusivity tensor and where $\bm{\nabla}_{\bm{X}}$ is the gradient vector with respect to the spatial coordinate ${\bm{X}}$. 
Note that $ {\bm{\mathfrak{D}}}$ comprises contributions from molecular diffusion, Taylor or shear dispersion, and turbulent mixing; we allow for shear in the in plane components of $ {\bm{\mathfrak{D}}}$ and further allow for dependence on ${\bm{X}}$ but we take the longitudinal component to be a purely unidirectional constant whose value is an experimentally determined average. 

The condition at the impermeable wall of the cylinder, $\partial{\Omega}$, reads
\begin{equation}
\label{cylinderBC}
\left({\bm{\mathfrak{D}}}\cdot\bm{\nabla}_{\bm{x}}{c}\right)\cdot\bm{n}_o=0, \quad {\bm{X}}\in\partial{\Omega},
\end{equation}
where $\bm{n}_o$ is the unit, outward-facing normal to $\partial{\Omega}$.  At the fluid-fluid interface which bounds the porous grains,  we enforce a mass balance between the inter- and intra-particle regions such that the mass transfer is determined by an empirical relation (discussed  in \S\ref{x-sec-ave-sect}). These conditions couple the inter-particle and intra-particle regions, thus before considering any conditions here, we present the problem in the intra-particle problem. 
\end{subequations}

\subsubsection{Intra-particle model}
As stated we consider a dual-porosity limit, \textit{i.e.}, there is a preferential flow path in the inter-particle domain. 
As such we assume that the flow, and consequently contaminant transport via advection, is negligible within the porous grains and so model the contaminant concentration ${c}_p^j$ within each grain via a diffusive transport equation 
\begin{equation}
\label{intra_raw}
    \frac{\partial{c}_p^j}{\partial{t}}=\bm{\nabla}_{\bm{x}}\cdot\left({\bm{\mathfrak{D}}}_p({\bm{x}})\cdot\bm{\nabla}_{\bm{x}}{c}_p^j\right),\quad{\bm{x}}\in{\omega}_f\, , \quad {\bm{X}}\in{\Omega}_p^j\, ,
\end{equation}
where ${\bm{\mathfrak{D}}}_p({\bm{x}})$ is an effective diffusivity tensor  and  $\bm{\nabla}_{\bm{x}}$ is the gradient vector with respect to the spatial coordinate ${\bm{x}}$. 
At the fluid-solid interface, $\partial {\omega}_s$, the adsorption reaction occurs 
\begin{equation}
\label{BC_ad_sa}
\int_{\partial {\omega}_s}\Big({\bm{\mathfrak{D}}}_p({\bm{x}})\cdot\bm{\nabla}_{\bm{x}}{c}_p^j\Big)\cdot\bm{n}_s \mathrm{d}S_x =
\frac{\partial{M}_p^j}{\partial {t}}\, , \quad {\bm{X}}\in{\Omega}_p^j\, , 
 \end{equation}
 where $\bm{n}_s$ is the outward facing unit normal to the solid in a porous grain,   $\mathrm{d}S_x$ is the area element with respect to $\bm{x}$ and we recall that ${M}_p^j({\bm{X}},{t})$ is the number of moles that have been adsorbed by the $j^\text{th}$ porous grain. Equation (\ref{BC_ad_sa}) is then a mass flux  balance through the surface ${\omega}_s$; $M_p^j$ will be specified based on the chemical behaviour of the system. 

At the fluid-fluid interface $\partial \omega$,  we  apply  a condition describing the total flux of mass $F_j$ entering the $j^\text{th}$ porous grain 
\begin{equation}
\label{BC_F_intra}
\int_{\partial {\omega}}\Big({\bm{\mathfrak{D}}}_p({\bm{x}})\cdot\bm{\nabla}_{\bm{x}}{c}_p^j\Big)\cdot\bm{n}_p \mathrm{d}S_x =
F_j, \quad {\bm{X}}\in{\Omega}_p^j, 
 \end{equation}
 where $\bm{n}_p$ is the outward-facing unit normal to the porous grain. 
We later specify the functional form of $F_j$. Note that the same  condition must apply in the interparticle region to ensure conservation of mass (see \S\ref{x-sec-ave-sect}). 

The aim of this model is to predict the macroscopic and filtration behaviour of materials. As such, we are not concerned with modelling the diffusive behaviour within a single grain but rather the effect of these grains on the macroscale filtration properties. Thus, we consider the intrinsic (fluid) average over a porous grain for a given quantity $\star=\star({\bm{x}},{\bm{X}},{t})$ 
    
\begin{equation}
\label{vol_ave}
\langle\star\rangle({\bm{X}},{t})\defeq\frac{1}{|{\omega}_f|}\int_{{\omega}_f}\star({\bm{x}},{\bm{X}},{t})\ \mathrm{d}V_x \equiv \frac{1}{\phi_p|{\omega}|}\int_{{\omega}_f}\star({\bm{x}},{\bm{X}},{t})\ \mathrm{d}V_x\, , \quad {\bm{X}}\in{\Omega}_p^j\, , 
\end{equation}
 for some $j$, where $|{\omega}_f|$ is the total fluid volume in a given porous grain and where $\mathrm{d}V_x$ is a volume element in the intra-particle fluid region, with respect to $\bm{x}$. 
 Taking the intrinsic average of Equation~(\ref{intra_raw}), applying the divergence theorem,  and recalling that ${c}_p^j \equiv 0$ for ${\bm{x}}\in {\omega}_s$, yields

\begin{align}
\label{intra_dev}
    \frac{\partial \langle {c}_p^j\rangle}{\partial {t}} &=
    \frac{1}{|{\omega}_f|}\int_{\partial \omega_f} \Big({\bm{\mathfrak{D}}}_p({\bm{x}})\cdot\bm{\nabla}_{\bm{x}}{c}_p^j\Big)\cdot \bm{n}_f \mathrm{d}S_x 
     \nonumber \\ 
    &\equiv  -\frac{1}{|{\omega}_f|}\int_{\partial \omega_s} \Big({\bm{\mathfrak{D}}}_p({\bm{x}})\cdot\bm{\nabla}_{\bm{x}}{c}_p^j\Big)\cdot \bm{n}_s \mathrm{d}S_x+
    \frac{1}{|{\omega}_f|}\int_{\partial \omega} \Big({\bm{\mathfrak{D}}}_p({\bm{x}})\cdot\bm{\nabla}_{\bm{x}}{c}_p^j \Big)\cdot \bm{n}_p \mathrm{d}S_x 
    \\
    &\label{itstoolatetothink} \equiv  
-\frac{1}{\phi_p|{\omega}|}\frac{\partial{M}_p^j}{\partial {t}}+\frac{F}{\phi_p|{\omega}|}\, , 
\quad {\bm{X}}\in{\Omega}_p^j\, , \nonumber %quad 
\end{align}
where $\bm{n}_f$ is the outward-facing unit normal to the fluid within a porous grain and where the last equality follows from the conditions~\eqref{BC_ad_sa} and~(\eqref{BC_F_intra}. Note that by taking the intrinsic average over each porous grain, $\langle{c}_p^j\rangle$ is now independent of ${\bm{x}}$. We define the \emph{extension} operator $ E[y^j]$ to be
\begin{equation}
    E[y^j]\defeq \left\{
    \begin{array}{lll}
       y^j\, , & \text{for } &\bm{X}\in\Omega_p^j\, , \quad j\in\{1,2,\cdots,N\}\, ,\\
       0\, , & \text{if }  &\bm{X}\in\Omega_i\, ,
    \end{array}
\right. 
\end{equation}
leading to the piecewise extended version of the intra-particle concentration field, $c_p$, the removal sink, $M_p$, and the mass flux, $F$, 
\begin{equation}
\label{c_p_defn}
     c_p(\bm{X},t)=E\left[\langle c_p^j\rangle\right]\, , \qquad M_p(\bm{X},t)=E\left[M_p^j\right]\, , \qquad F(\bm{X},t)=E\left[F_j\right]\, ,
\end{equation}
which are now defined over the whole domain $\Omega$.

% We define the intra-particle concentration field, $c_p$, via the piecewise function 
% \begin{equation}
% \label{c_p_defn}
%     c_p(\bm{X},t)\defeq \left\{
%     \begin{array}{lll}
%        \langle{c}_p^j\rangle\, , & \text{for } &\bm{X}\in\Omega_p^j\, , \quad j\in\{1,2,\cdots,N\}\, ,\\
%        0\, , & \text{if }  &\bm{X}\in\Omega_i\, .
%     \end{array}
% \right. %\bigcup_j\langle c_p^j\rangle, \quad \text{for all} 
% \end{equation}
% Note that we have set $c_p\equiv 0$ for $X\in\Omega_i$.
% Analogous to the definition of $c_p$ (Eq.~\ref{c_p_defn}), we define the removal  (sink) ${M}_p$  for all $\bm{X}\in \Omega$ via 
% \begin{equation}
%     \label{M_p_defn}
%     M_p(\bm{X},t) %= \bigcup_j M_p^j, \quad \text{for all} \quad \bm{X}\in\Omega_p.
%     %c_p(\bm{X},t)
%     \defeq \left\{
%     \begin{array}{lll}
%        M_p^j\, ,& \text{for } &\bm{X}\in\Omega_p^j\, , \quad j\in\{1,2,\cdots,N\}\, ,\\
%        0\, , & \text{if }  &\bm{X}\in\Omega_i\, .
%     \end{array}
% \right.
% \end{equation}
% and  the flux of mass (source) $F$ for all $\bm{X}\in \Omega$ via 
% \begin{equation}
%     \label{F_defn}
%     F(\bm{X},t) 
%     \defeq \left\{
%     \begin{array}{lll}
%        F_j\, ,& \text{for } &\bm{X}\in\Omega_p^j\, , \quad j\in\{1,2,\cdots,N\}\, ,\\
%        0\, , & \text{if }  &\bm{X}\in\Omega_i\, .
%     \end{array}
% \right.
% \end{equation}
To close the problem we must prescribe the behaviour of the removal ${M}_p$ and the mass flux $F$. This is carried out at throughout \S\ref{x-sec-ave-sect}.

\subsection{Model simplification}

\label{x-sec-ave-sect}
\begin{figure}[ht!]
    \centering    \includegraphics[width=.7\textwidth]{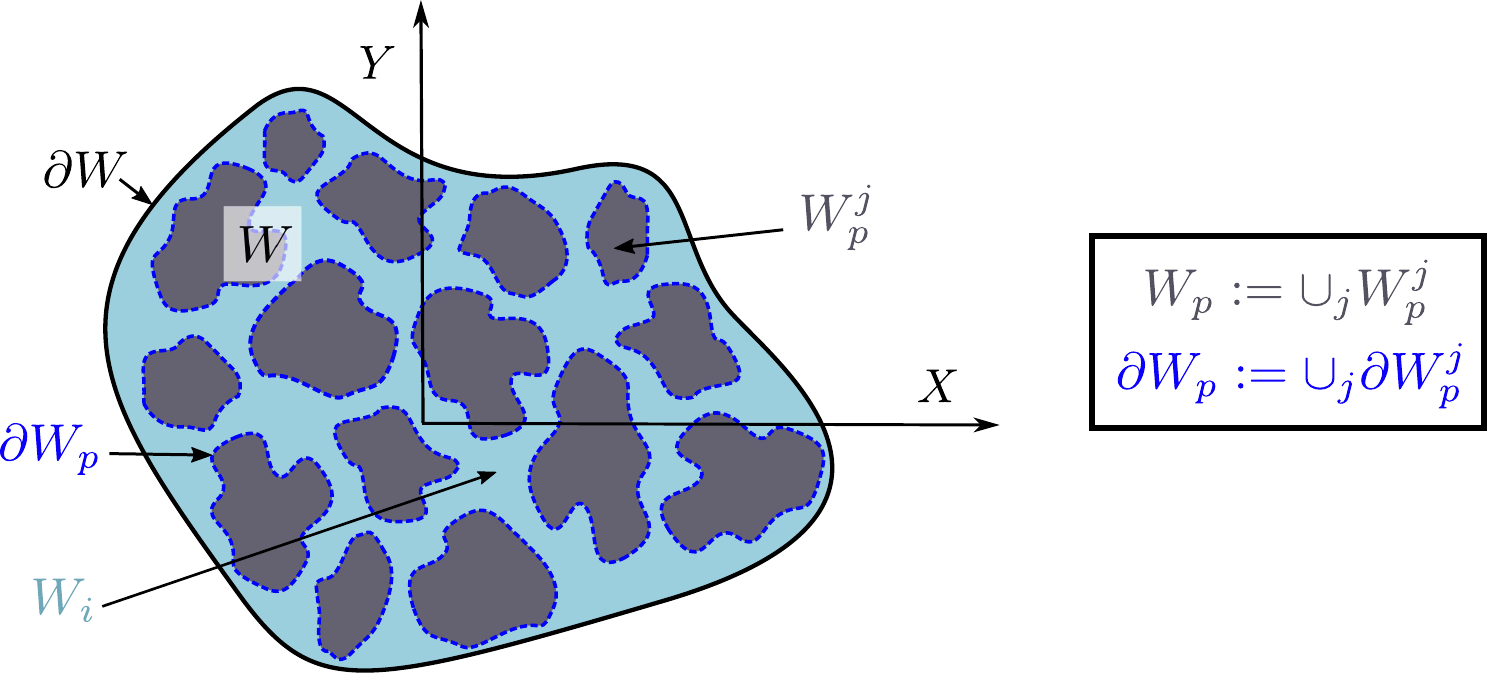}
    \caption{Cross-section of the filter spanned by the ${X}$-${Y}$ plane.}
    \label{cross-sec}
\end{figure}

The aim of this paper is to develop a simple model for for flow and adsorption in a dual-porosity material. 
Consequently,  we employ cross-sectional averaging over the macroscale which is inspired by an ensemble average \cite{Myer20a,Myer20b,Valverde2024}. This approach has previously shown good agreement with experimental data \cite{Valverde2024,Aguareles2023}. Mirroring the definitions of ${\Omega}$ we now define notation in an arbitrary two-dimensional cross-section of the filter  $Z=Z^\star\in\Omega$: the  inter-particle domain is denoted  ${W}_i(Z^\star)$, the union of the  distinct particle domains ( ${W}_p^j(Z^\star)$ ) 
is denoted ${W}_p(Z^\star)$, with the domain of the entire cross-section denoted ${W}$ with boundary $\partial W$ (Figure \ref{cross-sec}). Note that although ${W}_p$  and ${W}_i$ are functions of $Z^\star$, ${W}$  is independent of $Z^\star$ as the shape of the filter is constant. Note further that a consequence of the ensemble average is that $\phi\defeq |{\Omega}_i|/|{\Omega}| \equiv |{W}_i|/|{W}|$.  Finally, we define a set of integer numbers $A(Z^\star)\subset \{1,2,\dots, N\}$ where $j\in A(Z^\star)$ if and only if the plane determined by $Z=Z^\star$ intersects~$\partial\Omega_p^j$.  

Given any  function  $\triangle = \triangle({\bm{X}},{t})$  defined on ${\Omega}$ we define the cross-sectional average via 
\begin{equation}
\label{x-ave}
   \bar{\triangle}({Z},{t})\defeq \frac{1}{\phi|{W}|}\int_{{W}_i}\triangle({\bm{X}},{t})\ \mathrm{d}S_X  +  \frac{1}{(1-\phi)|{W}|}\int_{{W}_p}\triangle({\bm{X}},{t})\ \mathrm{d}S_X,
  % = \frac{1}{|{W}|} \left(\int_{{W}}\triangle({\bm{X}},{t})\ \mathrm{d}S
   \quad {Z}\in(0,L).
\end{equation}
where $\mathrm{d}S_X$ is a cross-sectional area element, with respect to $\bm{X}$. In particular, recall that ${c}\equiv 0$ for ${\bm{X}}\in{\Omega}_p$ %and that ${W}_i = \phi \ {W}$ 
so that 
\begin{subequations}
\begin{equation}
\bar{{c}}(Z,t) = \frac{1}{\phi|{W}|}\int_{{W}_i}c({\bm{X}},{t})\ \mathrm{d}S_X\, , \quad {Z}\in(0,L)\, . 
\end{equation}
Application of Equation~(\ref{x-ave}) to $c_p$, defined in Equation~(\ref{c_p_defn}), yields
\begin{equation}
\label{c_pbar_defn}
    \bar{c}_p(Z,t)   = \frac{1}{(1-\phi)|{W}|}\int_{{W}_p} c_p\ \mathrm{d}S_X  = \frac{1}{(1-\phi)|{W}|}\sum_{j\in A(Z)}\int_{{W}_{p}^j}\langle c_p^j \rangle \mathrm{d}S_X\, , \quad {Z}\in(0,L)\, .
\end{equation}
Similarly, we calculate the cross-sectional average mass flux between the inter- and intra-particle regions given by 
\begin{equation}
\label{F_bar_defn}
    \bar{F}(Z,t)   = \frac{1}{(1-\phi)|{W}|}\int_{{W}_p} F\ \mathrm{d}S_X\, , \quad {Z}\in(0,L)\, .
\end{equation}
Note that $\bar{c}(Z,t)$, $\bar{c}_p(Z,t)$ and $\bar{F}(Z,t)$, are now independent of $X,Y$ and $j$.  

At the fluid-fluid boundary we prescribe $F$ on $\partial{\omega}$ combining a mass balance and a Newton cooling law:
\begin{equation}
    \label{Fbar}
    {F}:= \int_{\partial \omega}  {k}_p(\bar{c}-\bar{c}_p)\ \mathrm{d} S_x \equiv |\partial \omega|{k}_p(\bar{c}-\bar{c}_p)\, ,
\end{equation}
where we have introduced the proportionality constant $k_p$ and the last equality follows  from the fact that $\bar{c}$ and $\bar{c}_p$ are independent of $\bm{x}$, by construction. Thus, Equation~(\ref{F_bar_defn}) simply yields 
\begin{equation}
\label{F_more}
    \bar{F}    \equiv F\, .
\end{equation}
 \end{subequations}
So that conservation of mass  requires 
% \begin{subequations}
\label{assumptions}
\begin{equation}
\label{assumption}
\int_{\partial{W}_{p}}\left({\bm{\mathfrak{D}}}^{X,Y}\cdot\bm{\nabla}^{X,Y}_{{\bm{X}}}{c}\right)\cdot\bm{n}_{p}^{X,Y}\mathrm{d}s_X\\  =  \frac{|W|(1-\phi)}{|{\omega}|}\bar{F},
\end{equation}
where
% \begin{equation}
%     {\bm{\mathfrak{D}}}^{X,Y}({\bm{X}})\defeq \begin{pmatrix}
% {D}_{11}({\bm{X}}) &  {D}_{12}({\bm{X}}) \\
% {D}_{21}({\bm{X}}) & {D}_{22}({\bm{X}}) 
% \end{pmatrix}\, ,
% \end{equation}
% \end{subequations}
$\bm{n}_{p}^{X,Y}$ is the  normal to $W_p$, $\bm{\nabla}^{X,Y}_{{\bm{X}}}\defeq\left(\displaystyle\frac{\partial}{\partial {X}},\displaystyle\frac{\partial}{\partial {Y}}\right)^\intercal$ and $\mathrm{d}s_X$ is the scalar line element with respect to $\bm{X}$. 
%Equation~(\ref{assumption}) is not intuitive and it results from the fact that to derive the governing equation for $\langle {c}_p^j\rangle$ (Equation~\ref{intra_dev}) we considered the diffusive flux to be in three dimensions, however to obtain a simple one dimensional equation we now assume the diffusive flux acts only in the plane. 
%A discussion showing this leads to conservation of mass for any cross-section is given below Equation~(\ref{dim_cp}). 

\subsubsection{Inter-particle}
Taking the cross-sectional  average~\eqref{x-ave} of Equation~\eqref{itsgonemidnight} yields
\begin{subequations}
\begin{equation}
\label{eek}
    \frac{\partial \bar{c}}{\partial{t}}+{v}\frac{\partial \bar{c}}{\partial{Z}} = \frac{1}{\phi|{W}|}\int_{{W}_i}\left(\bm{\nabla}_{{\bm{X}}}^{X,Y}\cdot\left( {\bm{\mathfrak{D}}}^{X,Y}\cdot\bm{\nabla}^{X,Y}_{{\bm{X}}}{c} \right) +\frac{\partial}{\partial{Z}}\left(D\frac{\partial{c}}{\partial{Z}}\right)\right)\mathrm{d}S_X, \quad {Z}\in(0,L)\, ,
\end{equation}
where the left-hand side  is a consequence of   our assumption that changes
to the volume and porosity of the solid structure are negligible and where we have decomposed the diffusive flux into the in plane and out of plane components. 
Noting $\partial{W}_i \equiv \partial{W}\cup\left(-\partial{W}_p\right)$, where the minus sign is used to preserve orientation, application of the Divergence theorem  on Equation~\eqref{eek}  gives 
\begin{align}
\begin{split}\label{eeeek}
    \frac{\partial \bar{c}}{\partial{t}}+{v}\frac{\partial \bar{c}}{\partial{Z}}-D\frac{\partial^2 \bar{c}}{\partial{Z}^2}  =& 
    \frac{1}{\phi|{W}|}\left[\int_{\partial{W}}\left( {\bm{\mathfrak{D}}}^{X,Y}\cdot\bm{\nabla}^{X,Y}_{{\bm{X}}}{c} \right)\cdot\bm{n}_o^{X,Y}\mathrm{d}s_X\right. \\
    &-\left. \int_{\partial{W}_p}\left( {\bm{\mathfrak{D}}}^{X,Y}\cdot\bm{\nabla}^{X,Y}_{{\bm{X}}}{c} \right)\cdot\bm{n}_p^{X,Y}\mathrm{d}s_X\right], \quad {Z}\in(0,L)\, ,
\end{split}\end{align}
where we have once again used our assumption that changes
to the volume and porosity of the solid structure are negligible and where $\bm{n}_o^{X,Y}$ is the in plane normal to the outer edge of the cylinder.
\end{subequations}
As we assume that the filter is a cylinder of arbitrary cross-sectional shape, we have that the condition~\eqref{cylinderBC} is equivalent to 
\begin{equation}
\label{cylinderBCinplane}
\left({\bm{\mathfrak{D}}}^{X,Y}\cdot\bm{\nabla}_{{\bm{X}}}^{X,Y}{c}\right)\cdot\bm{n}_o^{X,Y}=0 \quad {\bm{X}}\in\partial{W}. 
\end{equation}
%\begin{subequations}
    Thus, applying~(\ref{cylinderBCinplane})~and~(\ref{assumption}) to Equation~(\ref{eeeek}) yields the cross-sectionally averaged equation for the inter-particle concentration 
    \begin{equation}
    \label{final_dim_c}
     \frac{\partial \bar{c}}{\partial{t}}+{v}\frac{\partial \bar{c}}{\partial{Z}}-D\frac{\partial^2 \bar{c}}{\partial{Z}^2} = -\frac{{k}_p(1-\phi)}{\phi}\frac{|\partial \omega|}{|\omega|}\left(\bar{c} - \bar{c}_p\right), \quad {Z}\in (0,L),
    \end{equation}
    where the quantity $ |\partial{\omega}|/|{\omega}| $ is known as the specific surface; in the case of spherical grains this is three times the inverse of the radius of the sphere. 

\subsubsection{Intra-particle}
\label{M_p_m_p}
Currently, the governing equations for $\langle{c}_p^j\rangle$ and $M_p^j$ are given in Equations~\eqref{intra_dev} and \eqref{BC_ad_sa}, respectively, for a single arbitrary grain, the location of which depends on ${\bm{X}}$; now we take a cross-sectional average of these equations to determine their behaviour uniformly over the desired one-dimensional domain ${Z}\in{\Omega}$ --- that is, $Z\in[0,L]$. Firstly, we define $\bar{m}_p$ with units moles$/$kg, to be the amount of adsorbate per unit mass of adsorbent,
\begin{equation}
\label{m_p_def}
\bar{m}_p(Z) \defeq \frac{(1-\phi)\bar{m}_p}{ \rho_b |\omega|}\equiv \frac{1}{\rho_b  |W| |\omega|}\sum_{j\in A(Z)}\int_{{W}_p^j}  {M}_p^j(\bm{X},t) \mathrm{d}S_X,  
\end{equation}
where $\rho_b$ is the bulk density and is defined to be the initial mass of the adsorbent divided by the total volume of the filter. We choose these units for $\bar{m}_p$ since this is a standard quantity use in experimental papers, \textit{e.g.,} \citep{Valverde2024,Borba2008,Dantas2011,Shaf15}. Consider
\begin{equation}
\begin{split}
\label{dim_cp}
\frac{\partial \bar{c}_p}{\partial t} =  \frac{1}{(1-\phi)|W|}&\sum_{j\in A(Z)}\int_{W_p^j}\frac{\partial \langle c_p^j\rangle}{\partial t}\ \mathrm{d}S_X  \\ 
&=  \frac{1}{(1-\phi)|W|}\sum_{j\in A(Z)}\int_{W_p^j}\left( -\frac{1}{\phi_p|\omega|}\frac{\partial M_p^j}{\partial t}+ \frac{F_j}{\phi_p|\omega|}\right) \mathrm{d}S_X \\ 
&=  -\frac{\rho_b}{\phi_p(1-\phi)}\frac{\partial \bar{m}_p}{\partial t} + \frac{{k}_p}{\phi_p}\frac{|\partial \omega|}{|\omega|} \left(\bar{c}-\bar{c}_p\right) \quad \text{for} \quad Z\in(0,L), 
\end{split}
\end{equation}
where we have used  Equation~\eqref{intra_dev}, the definitions~\eqref{c_pbar_defn} and~\eqref{m_p_def} of $\bar{c}_p$ and $\bar{m}_p$, and $\bar{F}$  as defined in Equations~(\ref{F_bar_defn} --\ref{F_more}), and the fact that $|\partial \omega|$, $|\omega|$, $\phi$, $\phi_p$ are constant.

Now, we fix the macroscale function which describes the removal; we use the Sips model given by
\begin{equation}
\label{SIPs}
   \frac{\partial\bar{m}_p}{\partial {t}} =  {k}_+  \bar{c}_p^a (\bar{m}_{\text{max}}-\bar{m}_p)^b - {k}_{-} \bar{m}_p^b\, , \quad {Z}\in(0,L)\, ,
\end{equation}
where $a$ and $b$ are integer exponents, ${k}_{+}$ and ${k}_{-}$, are the adsorption and desorption rate constants, respectively, and $\bar{m}_{\text{max}}$ is the maximum amount of adsorbate that can attach to the adsorbent surface per mass of adsorbent. As is standard in literature surrounding chemistry and chemical engineering, the units of ${k}_{+}$ and ${k}_{-}$ are not the same. Note that, when $a=b=1$ Equation~(\ref{SIPs}) reduces to the classical Langmuir-removal model. 

At equilibrium, the concentration $\bar{c}$ is the same as the feed (or inlet) concentration, and the adsorbed amount, denoted $\bar{m}_e$ is  
% \begin{equation}
% \label{Sipsiso}
%     \bar{m}_e\defeq\displaystyle\frac{k_+^\frac{1}{b} c_{in}^\frac{a}{b}\bar{m}_{max}}{k_-^\frac{1}{b}+k_+^\frac{1}{b}c_{in}^\frac{a}{b}}\, .
% \end{equation}
\begin{equation}
\label{Sipsiso}
\bar{m}_e\defeq\displaystyle\frac{\bar{m}_{max}}{1+\left(\mathcal{K}c_{in}^a\right)^{-1/b}}\, ,
\end{equation}
where $\mathcal{K}\defeq k_+/k_-$ is known as the equilibrium constant. Equation  (\ref{Sipsiso}) this is known as the Sips isotherm.

\subsection{Boundary conditions}
We take the standard boundary and initial conditions \cite{Myer20a, Myer20b, Myer22,Valverde2024, Aguareles2023};  conservation of flux at the inlet (Dankwert's condition), a passive outflow condition at the outlet, and initial conditions that the material is contaminant free:
\begin{subequations}
\label{BC_dim}
    \begin{align}
\label{BC_1}
    {v}\bar{c}(0,t)-D\left.\frac{\partial \bar{c}}{\partial Z}\right|_{Z=0}={v}c_{in}\, , \quad &\text{for all } t\, ,\\
\left.\frac{\partial \bar{c}}{\partial Z}\right|_{Z=L}=0\, , \quad &\text{for all }t\, ,\\
    \bar{c}=\bar{c_p}=\bar{m}_p=0\, , \quad &\text{at } t=0. 
\end{align}
\end{subequations}
The asymmetry between the inlet and outlet conditions is explored and justified in \citet{Pear59}. 

\subsection{Non-dimensionalisation}
Our system is modelled by the coupled system of partial differential equations formed by Equations~\eqref{final_dim_c}, \eqref{dim_cp}, and \eqref{SIPs}, and subject to the boundary conditions~\eqref{BC_dim}. We non-dimensionalise  via the scalings 
\begin{subequations}  
\label{scaling}
\begin{equation}
 \bar{c} = c_{in} C, \quad \bar{c}_p = c_{in} C_p,  
 \quad  \bar{m}_p= \bar{m}_e m_p, \quad  {Z}=\mathcal{L}\zeta, \quad \text{and}\quad {t} = \mathcal{T}\tau\, ,
\end{equation}
where 
$c_{in}$ is the inlet concentration, $\bar{m}_e$ is defined in Equation~\eqref{Sipsiso}, $\mathcal{L}$ is the reaction length scale and $\mathcal{T}$ is the reactive timescale. 
The reaction length scale  is obtained via balancing advection with the rate of  removal ($\partial m_p/\partial t$) in Equations~\eqref{final_dim_c} and \eqref{dim_cp} and the reactive timescale is obtained via balancing the rate of  removal with  adsorption in Equation~(\ref{SIPs}) giving 
\begin{equation}
    \mathcal{L}\defeq \frac{{v}\mathcal{T}\phi c_{in}}{\rho_b \bar{m}_e}\quad \text{with} \quad \mathcal{T}\defeq \frac{\bar{m}_e^{1-b}}{{k}_+c_{in}^a}\, . 
\end{equation}
 \end{subequations}
We  denote the dimensionless length of the filter as $l\defeq{{L}}/{\mathcal{L}}$.
% \begin{equation}
% l\defeq\frac{{L}}{\mathcal{L}}.
% \end{equation}
Employing the scaling~(\ref{scaling}) on Equations~\eqref{final_dim_c}, \eqref{dim_cp}, \eqref{SIPs}, and the boundary conditions (\ref{BC_dim}) yields 
\begin{subequations}
\label{non-dim}
\begin{align}
\label{C_eqn}
    \mathrm{Da}\frac{\partial C}{\partial \tau} + \frac{\partial C}{\partial \zeta} - \mathrm{Pe}^{-1}\frac{\partial^2 C}{\partial \zeta^2} &= -\left(\alpha \frac{\partial C_p}{\partial \tau}+\frac{\partial m_p}{\partial \tau}\right)\, , \\
    \alpha\frac{\partial C_p}{\partial \tau} + \frac{\partial m_p}{\partial \tau} &= \beta\left(C - C_p\right)\, , \\ 
    \frac{\partial m_p }{\partial \tau}&=C_p^a\left(\mu-m_p\right)^b-(\mu-1)^b m_p^b\, ,
\end{align}\end{subequations}
subject to
\begin{subequations}
\label{non-dim-BC}
    \begin{align}
\label{BC_1a}
    C(0,\tau)-\Pe^{-1}\left.\frac{\partial C}{\partial \zeta}\right|_{\zeta=0}=1\, , \quad &\text{for all} \quad \tau>0\, , \\ \label{BC_2a}
\left.\frac{\partial C}{\partial \zeta}\right|_{\zeta=l}=0\, , \quad &\text{for all} \quad \tau>0\, ,\\
\label{IC_1}
    C(\zeta,0)=C_p(\zeta,0)=m_p(\zeta,0)=0\, , \quad &\text{for all} \quad \zeta\in(0,l)\, ,
\end{align}
\end{subequations}
where we have defined
\begin{align}\begin{split}
\label{mp_eqn}
    &\mathrm{Da}\defeq\frac{\mathcal{L}}{{v}\mathcal{T}}\, ,\quad \mathrm{Pe}\defeq \frac{{v}\mathcal{L}}{D}\, ,\quad \mu\defeq\frac{\bar{m}_\mathrm{max}}{\bar{m}_e}=1+\left(\mathcal{K}c_{in}^a\right)^{-1/b}\, ,\\ 
    &\alpha\defeq \frac{\phi_p\mathrm{Da}} {\varphi}\, ,\quad \beta\defeq{k}_p\frac{|\partial{\omega}|}{|{\omega}|} \frac{\mathcal{T}\mathrm{Da}}{\varphi}\, , \quad \text{with} \quad\varphi\defeq\frac{\phi}{1-\phi}\, .
\end{split}\end{align}
The parameters $\mathrm{Da}$ and $\mathrm{Pe}$ are the  Damk\"{o}hler  and P\'{e}clet numbers, respectively.

\subsection{Parameter values}
\label{sec:param}
The specific surface, $|\partial {\omega}|/|{\omega}|$, is the ratio of grain surface area to volume and satisfies that $|\partial {\omega}|/|{\omega}| \geq 3/R$, where we define $R$ to equal the shortest distance from the centre of a grain to the point on the surface that is furthest away from the centre (\textit{e.g.,} $R$ is the distance from the centre of a cube to one of the vertices). Note that this ratio is  minimised for a sphere. 

The coefficient of the rate of change of concentration within each porous grain, $\alpha$, can be thought of as the ratio of the flow timescale to the reaction timescale  within each porous grain, that is, an `internal' Damk\"{o}hler number. Provided the porous grains of adsorbent have a non-negligible fluid-fraction (\emph{i.e.}, $\phi_p\not\ll1$) and that the column is neither too packed ($1-\phi\not\ll1$) nor too empty (($\phi\not\ll1$)), then $\alpha=\ord{\mathrm{Da}}$. Alternatively, using the definition of $\mathcal{L}$ and $\mathrm{Da}$, $\alpha\equiv\phi_p(1-\phi)c_{in}/(\rho_b\bar{m}_e)$ is the ratio of the maximum amount of contaminant that can occupy all porous particles to the amount that the porous particles can adsorb per unit length. 

Experimental data (\textit{e.g.,} \cite{Aguareles2023,Myer22,Valverde2024,Myer20b,Myer20a, Sulaymon2014, Goeppert2014, Yousif2013}) suggests the following ranges for parameters: 
  $\phi_p, \phi\in(0.2, 0.8)$ which gives $\varphi\in(1/4,4) = \mathcal{O}(1)$; $\mathrm{Pe}^{-1}\in  (10^{-4}, 10^{-1})$; and $\alpha\sim\mathrm{Da}\in(10^{-6}, 10^{-3})$. Note that, in most cases  $\alpha,\Da\ll\mathrm{Pe}^{-1}\ll1$.
In practical situations, desorption should be at most comparable to adsorption, thus in general we expect $\mathcal{K}=k_+/k_-\gg1$. Since we have $\mu=1+/(\mathcal{K}c_{in}^a)^{1/b}$, this suggests that we expect $\mu=\ord{1}$ (with $\mu>1$).

% Also, the definition of $\mu$, with the fact that $\mathcal{K}=k_+/k_-$ should be large in any adsorption setting, gives that $\mu\geq 1$ but also $\mu=\ord{1}$ for the typical moderate values of $a,b$, \textit{i.e.}, $a,b<4$.

The parameter $\beta$ results from scaling the Newton Cooling condition \eqref{Fbar} with the removal inside porous grains. It can be  thought of a ratio of timescales, in particular the macroscale reactive timescale $\mathcal{T}$ over the timescale 
\begin{equation}
    \mathcal{T}_\beta\defeq \displaystyle\frac{\rho_b \bar{m}_e|{\omega}|}{c_{in}{k}_p(1-\phi)\left|\partial{\omega}\right|}\, ,
\end{equation} 
associated with the entry of the contaminant into a porous grain and its subsequent removal via adsorption. In what follows, we analyse the effects of varying $\beta$ on the effective performance of the filter.

\begin{description}
    \item[$\beta\gg1$:] in this regime contaminant almost instantly diffuses into the particle once it reaches the porous grain. It is possible to use effective parameters to obtain a governing set of equations that are structurally identical to the classical models. This allows describing the whole dual porosity medium by an equivalent single porosity medium with an effective porosity at leading order:
    \begin{subequations}
    \begin{align}
    \breve{\mathrm{Da}}\frac{\partial C}{\partial \tau} + \frac{\partial C}{\partial \zeta} - \mathrm{Pe}^{-1}\frac{\partial^2 C}{\partial \zeta^2} &= -\frac{\partial m_p}{\partial \tau}, \\
    \frac{\partial m_p }{\partial \tau}&=C^a\left(\mu-m_p\right)^b-(\mu-1)^b (m_p)^b,
\end{align}\end{subequations}
where $ \breve{\mathrm{Da}}\defeq \alpha+\mathrm{Da}$.
% \begin{equation}
%     \breve{\mathrm{Da}}\defeq \alpha+\mathrm{Da}. 
% \end{equation}
Thus, in this limit we replace the  `inter-particle' porosity $\phi$ with an effective porosity $\breve{\phi}$ given by $\breve{\phi}\defeq \phi(1-\phi_p)+\phi_p$.
% \begin{equation}
% \label{phi_breve}
%     \breve{\phi}\defeq \phi(1-\phi_p)+\phi_p\, . 
% \end{equation}
Note that, $\breve{\phi}\leq1$ because $(1-\phi_p)+\phi_p =1$ and $\phi\leq1$. This is often referred to as a model with reservoir effects \cite{royer2012time}. To exactly reproduce the classical model (with $\Phi = \phi$) we also enforce that we have impermeable  grains; thus, $\phi_p \to 0 \implies \alpha \to 0 $  so that $\breve{\mathrm{Da}}\to\mathrm{Da}$. This limiting process corresponds to converting the total adsorption internally within each grain to an effective removal on the surface of the porous grains. 
    \item[$\beta\sim1$:] In this limit the diffusion within the porous grain balances the rate at which contaminant advects to the surface of the porous grain in the interparticle domain. In this dual porosity limit the interplay between the inter- and intra-particle regions is crucial. 
    
    \item[$\beta\ll1$:]
    In this limit we have a very ineffective filter. At leading order we reproduce the single porosity limit \cite{royer2012time} but with no removal
    \begin{equation}
        \mathrm{Da}\frac{\partial C}{\partial \tau} + \frac{\partial C}{\partial \zeta} - \mathrm{Pe}^{-1}\frac{\partial^2 C}{\partial \zeta^2} = 0\, . 
    \end{equation}
This essentially corresponds to when the time taken to diffuse contaminant within the porous particles is much slower than the rate at which contaminant is advected to the porous grains. We also note that $\beta$ strongly depends on the porosity of the bulk material and on the specific surface of the grains. In particular, a large porosity ($\varphi$ large) and a small specific surface produce a small value of $\beta$.
\end{description}

Given that the limit $\beta \ll 1$ reduces to the classical model, and  $\beta \gg 1$ corresponds to negligible removal, in what follows we focus on the limit $\beta\sim1$.

\section{Solution methods}\label{sec:solutionmethods}

%\purple{I presume pseudo-spectral methods are well-documented - would this be better as being very short with details in Supplementary Info?}

\subsection{Numerical solution of the full system of partial differential equations}
In general, we solve the initial-boundary value problem (henceforth IBVP)~(\ref{non-dim},\ref{non-dim-BC}) numerically. We tackle this  IBVP  via a  combination of the method of lines and Chebyshev spectral collocation (\textit{i.e.}, Chebyshev pseudospectral methods) \citep{piche2007solving, trefethen2000spectral, bjornaraa2013pseudospectral, auton2018arteries}. That is, we discretise the spatial domain into the $N$ Chebyshev points and then approximate spatial derivatives using a dense  Chebyshev differentiation matrix \cite{auton2017arteries}; this approach allows for simple incorporation of the boundary conditions into the Chebyshev pseudospectral differentiation matrix. We subsequently  integrate the resulting system of differential algebraic equations (DAEs) in time with \verb|MATLAB|${\textsuperscript{\textregistered}}$ using \verb|ode15s| \citep{piche2009matrix, auton2018arteries}.  The Chebyshev pseudospectral method has proved to be more robust than classical finite volumes or finite differences; this is because a change at any collocation point  affects every other collocation point, while a change in a simple finite differences discretisation will only affect its neighbours. More detail on this approach  is given in \ref{Cheb_Num}. As we will discuss in \S\ref{TW}, due to the form of the full numerical solution,  we make a travelling wave assumption to reduce the problem to a system of first order ODEs with a boundary condition and an internal constraint for  which we also solve via  Chebyshev spectral collocation (see \S \ref{Num.sol.TW}).

\subsection{Steady-state solutions}
\label{sec:statSol}

At steady state, the system~(\ref{non-dim}) becomes
\begin{subequations}
\begin{align}
\frac{\partial C}{\partial \zeta}-\Pe^{-1}\frac{\partial^2C}{\partial \zeta^2} &= 0,
\\
C&=C_p,\\
\left(C_p\right)^a(\mu-m_p)^b-(\mu-1)^b(m_p)^b&=0,
\end{align}\end{subequations}
subject to the boundary conditions~(\ref{non-dim-BC}a,b). Hence the steady state solution is  simply
\begin{equation}\label{steadystatesoln}
    C(\zeta,\tau) = C_p(\zeta,\tau) = m_p(\zeta,\tau) \equiv 1.
    % \quad \text{for all} \quad\zeta\in[0,l],\tau\geq 0. 
\end{equation}
This trivial solution will be used in \S\ref{TW} as a boundary condition for a travelling wave approximation.

\subsection{Early time solution}
\label{sec:early}

Here, we investigate the form of the solutions at early times. Early time solutions are often useful to determine a starting point for a numerical scheme, particularly when the exact initial condition is discontinuous. Our simulations do not suffer from this problem, in which case the early time solution has been used to verify the numerical solution close to the initial condition.

We take $T=\tau/\delta$ and note that for small time we expect $\ord{1}$ changes in $\mathcal{C}$ and $\mathcal{C}_p$ but only $\mathcal{O}(\delta)$ changes in $m_p$.  That is, in a small time, we expect the filter to use up only a small fraction of its available adsorption sites which implies that  $m_p$ scales like $\delta\mathcal{M}$. Thus, we take $\mathcal{M}=m_p/\delta$ and the system~\eqref{non-dim} becomes
\begin{subequations}
\label{rescaled}
    \begin{align}
    \frac{\Da}{\delta}\frac{\partial C}{\partial T} + \frac{\partial C}{\partial \zeta} - \mathrm{Pe}^{-1}\frac{\partial^2 C}{\partial \zeta^2} &= -\left(\frac{\partial \mathcal{M}}{\partial T}+\frac{\alpha}{\delta}\frac{\partial C_p}{\partial T}\right), \\
    \frac{\alpha}{\delta}\frac{\partial C_p}{\partial T} + \frac{\partial \mathcal{M}}{\partial T} &= \beta\left(C - C_p\right), \\ 
  \frac{\partial \mathcal{M}}{\partial T}&=\left(C_p\right)^a\left(\mu-\delta\mathcal{M}\right)^b-(\mu-1)^b \left(\delta\mathcal{M}\right)^b,
\end{align}
\end{subequations}
subject to the boundary and initial conditions~\eqref{non-dim-BC}. Then, for any $\delta$ satisfying $\Da\ll\delta\ll 1$ and $\Pe\ll 1$ the leading order terms ($C^{(0)}$, $\mathcal{M}^{(0)}$ and $C_p^{(0)}$) are found to satisfy
\begin{subequations}\label{smalltime_eqs}
    \begin{align}
        \frac{\partial C^{(0)}}{\partial \zeta}  &= -\frac{\partial \mathcal{M}^{(0)}}{\partial T}\, ,\\
        \frac{\partial \mathcal{M}^{(0)}}{\partial T} &= \beta\left(C^{(0)} - C_p^{(0)}\right)\, ,\\
        \frac{\partial \mathcal{M}^{(0)}}{\partial T}&=\mu^b\left(C_p^{(0)}\right)^a\, ,
    \end{align}
\end{subequations}
subject to 
\begin{align}\label{smalltime_BCs}
    \text{(i)} \quad C^{(0)}(0,T)=1\, ,\quad \text{and} \quad \text{(ii)} \quad\left.\frac{\partial C^{(0)}}{\partial \zeta}\right|_{\zeta=l}=0\, .
\end{align}
Rearranging the early-time system~(\ref{smalltime_eqs}) gives
\begin{subequations}
\label{34}
\begin{align}
    &\frac{\partial C_p^{(0)} }{\partial\zeta}\left(1 + \frac{\mu^b}{\beta}a (C_p^{(0)})^{a-1}\right)=-C_p^{(0)}\mu^b,\label{eq:cp0}\\
    & C^{(0)} = \left(1+\frac{\mu^b}{\beta}\right)C_p^{(0)},\label{eq:c0}\\
    & \frac{\partial\mathcal{M}^{(0)}}{\partial T}=(C_p^{(0)})^a\mu^b.
\end{align}    
\end{subequations}
Equation \eqref{eq:cp0} can be solved explicitly, which determines $C_p^{(0)}$ and $C^{(0)}$ via Equation~\eqref{eq:c0}. 

For the case where $a=1$,  the general solution to the problem~(\ref{smalltime_BCs},\ref{34})  is 
\begin{equation}
    C_p^{(0)}(\zeta,T)= \varphi(T) e^{-\beta\mu^b\zeta/(\beta+\mu^b)}\, ,\qquad C^{(0)}= \frac{\beta+\mu^b}{\beta}\varphi(T)e^{-\beta\mu^b\zeta/(\beta+\mu^b)}\, .
\end{equation}
Note that there is only one free parameter, the time-dependent function $\varphi(T)$, however,  $C(\zeta, T)$  is subject to the two boundary conditions~\eqref{smalltime_BCs}. Thus, it is not possible to find a function $C(\zeta, T)$ that satisfies both boundary conditions; it is natural to choose the inlet condition~(\ref{smalltime_BCs}.i) since the process is driven by the input at $\zeta =0$. Further, imposing the boundary condition~(\ref{smalltime_BCs}.ii) would yield the trivial solution $C^{(0)}= C_p^{(0)}=\mathcal{M}^{(0)}\equiv 0$. Therefore, the solution of Equation~\eqref{smalltime_eqs} subject to~(\ref{smalltime_BCs}.i) gives
\begin{equation}
     C_p^{(0)}  =\frac{\beta}{\beta+\mu^b}  e^{-\beta\mu^b\zeta/(\beta+\mu^b)}\, , \quad
      C^{(0)}  = e^{-\beta\mu^b\zeta/(\beta+\mu^b)}\, ,\quad
      \mathcal{M}^{(0)}  = \frac{\beta T}{\beta+\mu^b}  e^{-\beta\mu^b\zeta/(\beta+\mu^b)}\, .
\end{equation}
We note that ,
\begin{equation}
\label{eq:bcInf}
\lim_{\zeta\to\infty} C^{(0)}=\lim_{\zeta\to\infty} \pad{C^{(0)}}{\zeta} =0\, .    
\end{equation}
Therefore, in an infinite filter the boundary condition at the outlet would also be satisfied. Further, Equation~\eqref{eq:cp0} shows that for any value of $a\geq 1$ smooth solutions satisfy conditions~\eqref{eq:bcInf}. In a filter of finite length the boundary condition at the outlet is satisfied by means of a trivial boundary layer (that is $C^{(0)}=$constant) of length $\delta$.

\subsection{Travelling wave approximation}
\label{TW}
The numerical solution to the IBVP~(\ref{non-dim},\ref{non-dim-BC})  implies that for certain parameter regimes and sufficiently large time, $C$, $C_p$ and $m_p$ behave like  travelling waves (see Figures \ref{fig:numerics_vs_TW1} and \ref{fig:numerics_vs_TW2}, and related discussion in \S\ref{sec:num_plots}). Thus, we introduce a travelling wave coordinate
\begin{equation}
\label{eq:TW-transf}
    \eta\defeq \zeta-S(\tau)\, . 
\end{equation}
We fix the location of the wavefront via 
\begin{equation}\label{TW:FrontPositionDef}
    S(\tau) = \dot{S}(\tau-\tau_{1/2})+l\, ,
\end{equation}
where  the wave velocity $\dot{S} \defeq \dd S/\dd\tau$ is a constant  to be determined during the solution process and $\tau_{1/2}$ is the (dimensionless) time at which the concentration at the outlet is exactly half that at the inlet (thus termed halftime), that is,  
\begin{equation}
\label{eq:C1/2}
    C(l,\tau_{1/2})=\frac{1}{2}\, .  
\end{equation}
We define 
\begin{equation}
    C(\zeta,\tau)\defeq\mathcal{C}(\eta), \quad  C_p(\zeta,\tau)\defeq\mathcal{C}_p(\eta),\quad \text{and}\quad m_p(\zeta,\tau)\defeq\mathcal{M}_p(\eta),
\end{equation}
so that the sytem~\eqref{non-dim} becomes 
\begin{subequations}
\label{eq:TW}
\begin{align}
\label{29a}
    \left(1-\Da\dot{S}\right) \frac{\dd\mathcal{C}}{\dd\eta} &= \Pe^{-1} \frac{\dd^2\mathcal{C}}{\dd\eta^2}+\alpha\dot{S}\frac{\dd\mathcal{C}_p}{\dd\eta}+\dot{S}\frac{\dd\mathcal{M}_p}{\dd\eta}\,\\
    \label{29b}
    \beta\left(\mathcal{C}-\mathcal{C}_p\right) &= -\dot{S}\left(\alpha\frac{\dd\mathcal{C}_p}{\dd\eta}+ \frac{\dd\mathcal{M}_p}{\dd\eta}\right)\, ,\\
    \label{29c}
    -\dot{S} \frac{\dd\mathcal{M}_p}{\dd\eta} &= \mathcal{C}_p^a\left(\mu-\mathcal{M}_p\right)^b-(\mu-1)^b\mathcal{M}_p^b\, ,
    \end{align}
\end{subequations}
and where the condition~\eqref{eq:C1/2} fixes 
\begin{equation}
   \label{eq:C1/2b}
   \mathcal{C}(0)=1/2.
\end{equation}
For the travelling wave approximation we need to consider filters of an infinite length. We thus need to impose boundary conditions at plus and minus infinity. Due to the transformation~\eqref{eq:TW-transf}, the limit $\eta\to-\infty$ corresponds to $\tau\to+\infty$
hence the steady state solution~\eqref{steadystatesoln} holds and we impose 
\begin{subequations}\label{eq:BCpluminusInf}
    \begin{equation}
    \label{eq:BCminusInf}
        \lim_{\eta\to-\infty} \mathcal{C}=\lim_{\eta\to-\infty} \mathcal{C}_p =\lim_{\eta\to-\infty}\mathcal{M}_p = 1\, , \quad \lim_{\eta\to-\infty} \frac{\dd\mathcal{C}}{\dd\eta}=0. 
    \end{equation}
These conditions are a natural choice since they correspond to a steady flow of pollutant sufficiently far upstream and that the filter is fully saturated there.  %Note that these conditions are also consistent with the inlet condition (Eq.~\ref{BC_1a}), that is, we have fixed $\zeta = 0$ and taken the limit as $\tau\to\infty$. 
 As $\eta\to\infty$ one expects the fluid to be clean and the filter fresh, thus we impose 
    \begin{equation}
    \label{eq:BCplusInf}
        \lim_{\eta\to\infty} \mathcal{C}=\lim_{\eta\to\infty} \mathcal{C}_p =\lim_{\eta\to\infty}\mathcal{M}_p = \lim_{\eta\to\infty} \frac{\dd\mathcal{C}}{\dd\eta}=0. 
    \end{equation}
      \end{subequations}

By integrating Equation~\eqref{29a} and imposing the conditions~\eqref{eq:BCplusInf} in the limit $\eta\to\infty$ ,  we find
\begin{subequations}
\begin{equation}
\label{11storder}
\left(1-\Da\dot{S}\right) \mathcal{C} = \Pe^{-1} \nd{\mathcal{C}}{\eta}+\dot{S}(\alpha\mathcal{C}_p+\mathcal{M}_p).    
\end{equation}
Imposition of the conditions~\eqref{eq:BCminusInf} determines the wave velocity $\dot{S}$:
\begin{equation}
    \label{eq:vel}
    \dot{S} =\frac{1}{1+\alpha+\Da}.
\end{equation}
\end{subequations}

In general, we solve Equations~\eqref{29b}, \eqref{eq:C1/2b}, and  \eqref{11storder} numerically, assuming that $\eta\in[-B, B]$, where $B$ is some large number, subject to $\mathcal{M}_p(B) = 0 $. Details of the numerical scheme are given in \S\ref{Num.sol.TW}. However, by taking advantage of the small parameters in the system, for certain combinations of $a$ and $b$, we are able to determine approximate analytical solutions.
Recall that for the physically relevant parameter regimes of the present study 
 $\alpha, \mathrm{Da}, \Pe^{-1} \ll1$. Neglecting  these terms, Equation~(\ref{eq:vel})   simply becomes  $\dot{S}=1$. 
 Further, Equation~(\ref{11storder}) simply yields 
 $\mathcal{C}=\mathcal{M}_p$. Thus Equations~\eqref{29b}, \eqref{eq:C1/2b}, and  \eqref{11storder} reduce to a first order ordinary differential equation for $\mathcal{C}(\eta)$ 
 while $\mathcal{C}_p(\mathcal{C})$ is determined implicitly via Equation~(\ref{29c})
\begin{subequations}
\label{eq:C}
    \begin{align}
        \label{30b}
        \frac{\dd \mathcal{C}}{\dd\eta} &= -\beta\left(\mathcal{C}-\mathcal{C}_p(\mathcal{C})\right) , \\
        \label{30c}
        \beta\left(\mathcal{C}-\mathcal{C}_p(\mathcal{C})\right)  &= \mathcal{C}_p^a(\mathcal{C})\left(\mu-\mathcal{C}\right)^b-(\mu-1)^b\mathcal{C}^b \, .
    \end{align}
    \end{subequations}
The following proposition establishes the conditions under which the system~\eqref{eq:C} provides a unique solution satisfying~\eqref{eq:C1/2b} and \eqref{eq:BCpluminusInf}, that is
\begin{equation}
\label{eq:CBC}
    \mathcal{C}(0)=1/2\, ,\quad \lim_{\eta\to -\infty} \mathcal{C} = 1\, , \quad \lim_{\eta\to \infty} \mathcal{C} = 0\, .
\end{equation}
\begin{prop}\label{prop:TW} Given $\mu,\beta, c_0\in\mathbb{R}$ such that $\mu>1$, $\beta>0$, and $0<c_0<1$, for any $a,b\in \mathbb{N}$ such that $a\leq b$, the initial value problem given by Equation \eqref{30b} along with the initial condition
\begin{equation}
    \label{eq:CCpIC2}
     \mathcal{C}(0) = c_0\in(0,1)\, ,
\end{equation}
where $\mathcal{C}_p(\mathcal{C})$ is implicitly determined by \eqref{30c}
is well posed and it has a unique decreasing solution, $\mathcal{C}(\eta)$ satisfying the conditions at plus and minus infinity provided in \eqref{eq:CBC}.
\end{prop}
The proof of this proposition is provided in~\ref{proofs}.

Proposition~\ref{prop:TW} states the existence of travelling waves as solutions of the reduced system~\eqref{eq:C} when $a\leq b$, but further analysis is needed to study cases where $a>b$. However, $a\leq b$ is the physically relevant case in adsorption settings. The physical reaction associated with an $(a,b)$ integer pair is 
\begin{equation}
\label{chemistry}
    a\mathcal{A} + b\mathcal{B} \xrightleftharpoons[k_-]{k_+} \text{Adsorbate}\, ,
\end{equation}
where $\mathcal{A}$ and $\mathcal{B}$ refer to the pollutant molecules and adsorbent sites, respectively. Although the case $a>b$ may appear an advantageous situation (each adsorbent site $\mathcal{B}$ is capable of capturing $a/b$ pollutant molecules $\mathcal{A}$), the equilibrium of this reaction provides a very slow increase of the adsorbed fraction at equilibrium $\bar{m}_e$ with respect to the inlet concentration. Thus, even high concentrations result in a scarce adsorption capacity. The cases in which $a>b$, provide convex isotherms that are regarded as unfavourable adsorption cases \cite{McCabe1993}.

In what follows we restrict ourselves to $a\leq b$. Analytical solutions can be determined for certain values of $a$ and $b$. When such solutions can be found, they may be determined via separation of variables in Equation~\eqref{30b}, which yields
\begin{equation}
     \label{integral_a=1}
     -\beta\eta = \int_{1/2}^\mathcal{C}  \frac{\dd u}{u-\mathcal{C}_p(u)} = \int_{1/2}^\mathcal{C}  f(u) \dd u=F(\mathcal{C}) - F\left(1/2\right)\, ,
     \end{equation}
where $F$ is defined by
\begin{equation}
    \label{eq:FC}
 F(\mathcal{C})=   \int_{0}^\mathcal{C}  f(u) \dd u\, ,
\end{equation}
and depends on the Sips exponents, $a$ and $b$. 
In what follows we consider some specific combinations of $(a,b)$ that are commonly found in the adsorption literature \cite{Aguareles2023}.
\begin{table}
    \centering
\begin{align*}
%a=1
\\
\begin{array}{|c|l|}
    \hline
     &\\
   a=1 \, &F(\mathcal{C}) \\
      &\\
    \hline\hline\hline
      &\\
    b=1 &(\beta+\mu)\ln(\mathcal{C})+ (1-\beta-\mu)\ln(1-\mathcal{C}) \\
      &\\
    \hline
      &\\
    b=2 &\left(\displaystyle\frac{\beta}{\mu^2}+1\right)\ln(\mathcal{C})- \displaystyle\frac{\beta+(\mu-1)^2}{\mu^2-1}\ln(1-\mathcal{C})+\displaystyle\frac{\beta+\mu^2(\mu-1)^2}{\mu^2(\mu^2-1)}\ln(\mu^2-\mathcal{C}) \\
      &\\
       \hline
     &\\
    b=3,\ \mu\in(1,4)& \left(\displaystyle\frac{\beta}{\mu^3}+1\right)\ln(\mathcal{C})-A_1\ln(1-\mathcal{C}) \, \, 
    \\
    &  \\
    & \hspace{1.7cm} +A_2\ln\left[(\mathcal{C}-q_1)^2+q_2^2\right]  -  \displaystyle\frac{A_3}{q_2}\arctan\left(\displaystyle\frac{\mathcal{C}-q_1}{q_2}\right)\\
     &\\
        & \text{where, $q_1$ and $q_2$ are defined in Equation~(\ref{q_i}) and where } \\
      &\\
        &A_1\defeq 
       \displaystyle\frac{\beta+(\mu-1)^3}{(\mu-1)^2(2\mu+1)}, \ \ 
        A_2\defeq 
        \displaystyle\frac{\mu^3(\mu-1)^3+\beta(2q_1-1)}{2\mu^3(\mu-1)^2(2\mu+1)}\\
         &\\
         & \text{and}\quad A_3\defeq \displaystyle\frac{\beta[3-\mu(4-\mu)(\mu-1)]+3(\mu-1)^3\mu^3}{2\mu(\mu-1)(2\mu+1)}\\
         & \\
    \hline
    & \\
    b=3, \ \mu=4 & 
    \displaystyle\frac{\beta+64}{64}\ln(\mathcal{C})-\displaystyle\frac{\beta+27}{81}\ln(1-\mathcal{C})+\displaystyle\frac{1728-17\beta}{5184}\ln(\mathcal{C}+8)+\displaystyle\frac{\beta+1728}{72(\mathcal{C}+8)} \\ 
     & \\ 
      \hline
      & \\
    b=3, \ \mu>4 &  \left(\displaystyle\frac{\beta}{\mu^3}+1\right)\ln(\mathcal{C}) - A_1\ln(1-\mathcal{C})+\left[
   \displaystyle\frac{\beta +(\mu-n_+)^3}{n_+(1-n_+)(n_+-n_-)}
   \right]\ln(\mathcal{C}-n_+)  \\
        & \\
     & \hspace{3.6cm} +  \left[
    \displaystyle\frac{\beta +(\mu-n_-)^3}{n_-(1-n_-)(n_--n_+)}
    \right]\ln(\mathcal{C}-n_-)
    \\
     & \\
      & \text{where $A_1$
      is defined above, in the case $b=3$, $\mu\in(1,4)$, and where $n_\pm$ are}\\
       & \text{defined in Equation~\eqref{n_i}.}\\ 
        &\\ 
     \hline
    \end{array}
\end{align*}
    \caption{Summary of expressions for $ F(\mathcal{C})$ defined in Equation~\eqref{eq:FC}, for different values of $a,b$.}
    \label{tab:F}
\end{table}

\subsubsection{Explicit solutions for typical parameter values}
\label{sec:exp}
When $a=1$, Equation~(\ref{30c}) can be rearranged as follows
\begin{equation}
\label{Cp}
    \mathcal{C}_p = \frac{(\mu-1)^b\mathcal{C}^b+\beta\mathcal{C}}{\beta+(\mu-\mathcal{C})^b}, 
\end{equation}
which we use to write Equation~\eqref{30b} as
\begin{equation}
\label{eq:C1}
    \frac{\dd\mathcal{C}}{\dd\eta}=-\beta \frac{\mathcal{C}(\mu-\mathcal{C})^b-(\mu-1)^b\mathcal{C}^b}{\beta+(\mu-\mathcal{C})^b}\, ,
\end{equation}
with implicit solution
\begin{equation}
     \label{eq:sol_a=1}
     -\beta\eta = \int_{1/2}^\mathcal{C} \frac{\beta+(\mu-u)^b}{u(\mu-u)^b-(\mu-1)^b u^b}\dd u=F(\mathcal{C}) - F\left(1/2\right)\, .
     \end{equation}
The evaluation of the integral when $b=1$ and $b=2$ is straightforward. 
When $b=3$, the form of the solution depends on   the number of real roots of the polynomial in the denominator of Equation~\eqref{eq:sol_a=1}:
\begin{equation}
    \label{eq:p_b=3}
   p(\mathcal{C})\defeq\mathcal{C}(\mu-\mathcal{C})^3-(\mu-1)^3 \mathcal{C}^3\, .
\end{equation}
When $\mu\in(1,4)$, $p(\mathcal{C})$ has only two real roots ($\mathcal{C}=0,1$), so we decompose it as follows
\label{B<4}
\begin{equation}
 p(\mathcal{C})   = \mathcal{C}(1-\mathcal{C})((\mathcal{C}-q_1)^2+q_2^2)
\end{equation} 
where
\begin{equation}
\label{q_i}
    q_1\defeq\frac{\mu^2(3-\mu)}{2}, \quad \text{and}\quad q_2^2\defeq  \mu^3 - \frac{\mu^4(3-\mu)^2}{4}\equiv\frac{\mu^3(4-\mu)(\mu-1)^2}{4}\, .
\end{equation} 
There exists a critical value of $\mu$, ($\mu=4$), such that $p(\mathcal{C})$ has three real roots, one of which has a double multiplicity; in this case $p(\mathcal{C})$ can be expressed as 
\begin{equation}
\label{B=4}
    p(\mathcal{C})= \mathcal{C}(1-\mathcal{C})(\mathcal{C}+8)^2. 
\end{equation} 
When $\mu>4$, $p(\mathcal{C})$ has four distinct real roots ($\mathcal{C}=0,1$), so we decompose it as follows
\begin{subequations}   
\label{B>4}
\begin{equation}
p(\mathcal{C})= \mathcal{C}(1-\mathcal{C})(\mathcal{C}-n_+)(\mathcal{C}-n_-)\, ,
\end{equation} 
with
\begin{equation}
\label{n_i}
    n_\pm\defeq \frac{\mu^2(3-\mu)\pm\sqrt{\mu^4(3-\mu)^2-4\mu^3}}{2}. 
\end{equation} 
\end{subequations}
Table \ref{tab:F} summarises the expressions for $F(\mathcal{C})$ for the above combinations of $a$ and $b$.

\subsubsection{Numerical solutions for travelling wave}
\label{Num.sol.TW}
 Equations~\eqref{29b}, \eqref{eq:C1/2b}, and \eqref{11storder} define a system of three coupled first order ODEs subject to one constraint on $\mathcal{C}$ at $\eta =0$. To solve this numerically we must approximate the infinite domain with a large but finite domain; as such we take $\eta\in[-B,B]$, where $B$ is a large constant.  We choose $B$ to be sufficiently large so that  
 \begin{equation}
 \nd{\mathcal{C}}{\eta}\, ,\   \nd{\mathcal{C}_p}{\eta}\, ,\  \nd{\mathcal{M}_p}{\eta} \to 0 \qquad \text{as}\qquad  \eta\to\pm B,
 \end{equation} 
 \textit{ie.,} the solution profiles are clearly asymptoting to zero at $\eta=-B$ and one at $\eta=B$. 
 In particular, in Figures \ref{fig:numerics_vs_TW1} and \ref{fig:numerics_vs_TW2} we take $B=75$.
In addition to  Equations~(\ref{29b},\ref{eq:C1/2b},\ref{11storder}), we must also impose a condition at $\eta=B$, that is, we enforce $\mathcal{M}_p(B)=0$.   
To solve this numerically, we use a direct method based on Chebyshev spectral collocation;  we discretise the spatial domain into an odd number of collocation points (to allow for easier enforcement of the condition at $\eta=0$) and use a dense Chebyshev differentiation  matrix to discretise the derivatives and determine an objective function. We solve
this nonlinear system via Newton iteration updating the solution via the analytically determined  Jacobian (see  \ref{Cheb_Num}) and terminating once the objective function is sufficiently small.

\subsubsection{Breakthrough curves}
When $\Da$ and $\alpha$ are neglected, Equation~\eqref{eq:vel} reduces to $\dot{S} = 1$. Therefore, the solution for any $a$, $b$ combination is given implicitly for $C(\zeta,\tau)$ by 
\begin{subequations}
 \label{solns}
\begin{equation}
 \label{solC}
     -\beta\left[\zeta-l-\left(\tau-\tau_{1/2}\right)\right]  = F({C}) - F\left(\frac{1}{2}\right),
 \end{equation}
  for $m_p(Z,t)$ by 
 \begin{equation}
 \label{solmp}
     -\beta\left[\zeta-l-\left(\tau-\tau_{1/2}\right)\right]  = F({m_p}) - F\left(\frac{1}{2}\right)
 \end{equation}
 and implicitly for $C_p(\zeta,\tau)$ via substitution of $C(\zeta,\tau)$ in Equation~\eqref{Cp}. 
 \end{subequations}
It is straightforward to find the breakthrough curves by evaluating Equation~\eqref{solC} at $\zeta = l$. In particular, the breakthrough curves that correspond to the analytical solutions presented in Table \ref{tab:F}, read:
\begin{itemize}
    \item $a=1, b=1$:
    \begin{equation}
        \label{eq:BreCu11}
         -\beta(\tau-\tau_{1/2})  = (\beta+\mu)\ln(2\mathcal{C})+ (1-\beta-\mu)\ln\left(2\left(1-\mathcal{C}\right)\right)\, ,
    \end{equation}
    \item $a=1, b=2$:
    \begin{equation}
    \begin{split}
        \label{a1b2C}
         -\beta(\tau-\tau_{1/2})  =& \left(\displaystyle\frac{\beta}{\mu^2}+1\right)\ln(2\mathcal{C})- \displaystyle\frac{\beta+(\mu-1)^2}{\mu^2-1}\ln\left(2\left(1-\mathcal{C}\right)\right)\\
         &+\displaystyle\frac{\beta+\mu^2(\mu-1)^2}{\mu^2(\mu^2-1)}\ln\left(\frac{\mu^2-\mathcal{C}}{\mu^2-1/2} \right) 
    \end{split}
    \end{equation}
    \item $a=1, b=3$: 
    
    If $\mu\in(1,4):$
    \begin{equation}
        \label{eq:BreCu13A}
        \begin{split}
            -\beta(\tau-\tau_{1/2})  =& -\left(\displaystyle\frac{\beta}{\mu^3}+1\right)\ln(2\mathcal{C})-A_1\ln\left(2\left(1-\mathcal{C}\right)\right)+A_2\ln\left(\frac{(\mathcal{C}-q_1)^2+q_2^2}{(1/2-q_1)^2+q_2^2}\right)  \\
            & -  \displaystyle\frac{A_3}{q_2}\left(\arctan\left(\displaystyle\frac{\mathcal{C}-q_1}{q_2}\right)+ \arctan\left(\displaystyle\frac{1/2-q_1}{q_2}\right)\right)\, ,
        \end{split}
    \end{equation}
    where $A_1, A_2, A_3, q_1, q_2,$, as defined in Table \ref{tab:F}, are constants that depend on $\mu$ and $\beta$,

    If $\mu = 4$:
    \begin{equation}
        \label{eq:BreCu13B}
        \begin{split}
            -\beta(\tau-\tau_{1/2})  =&\left(\displaystyle\frac{\beta}{\mu^3}+1\right)\ln(2\mathcal{C}) - A_1\ln\left(2\left(1-\mathcal{C}\right)\right)\\
            &+\left(
   \displaystyle\frac{\beta +(\mu-n_+)^3}{n_+(1-n_+)(n_+-n_-)}
   \right)\ln\left(\frac{\mathcal{C}-n_+}{1/2-n_+}\right) ,
        \end{split}
    \end{equation}

    If $\mu>4$, 
\begin{equation}
        \label{eq:BreCu13C}
        \begin{split}
            -\beta(\tau-\tau_{1/2})  =&\left(\displaystyle\frac{\beta}{\mu^3}+1\right)\ln(2\mathcal{C}) - A_1\ln\left(2\left(1-\mathcal{C}\right)\right)\\
            &+\left(
   \displaystyle\frac{\beta +(\mu-n_+)^3}{n_+(1-n_+)(n_+-n_-)}
   \right)\ln\left(\frac{\mathcal{C}-n_+}{1/2-n_+}\right) ,
        \end{split}
    \end{equation}
   where $A_1$ is defined in Table \ref{tab:F} and $n_\pm$ are defined in Equation~\eqref{n_i}.
\end{itemize}

\subsection{Comparison with numerical simulations}
\label{sec:num_plots}

In Figures \ref{fig:numerics_vs_TW1} and~\ref{fig:numerics_vs_TW2} we compare the numerical solution of the full IBVP~(\ref{non-dim},\ref{non-dim-BC}), the numerical solutions of the travelling wave approximation~\eqref{eq:TW}, and the analytical solutions derived in \S\ref{sec:exp} after setting $\alpha,\Da,\Pe^{-1}=0$ in the system~\eqref{eq:TW}. The parameter values for Figure \ref{fig:numerics_vs_TW1} have been chosen to be representative of the experimental data in \S\ref{HgII}, while those for Figure \ref{fig:numerics_vs_TW2} are representative of the experimental data in \S\ref{CO2}.

% Figures \ref{fig:numerics_vs_TW1} and~\ref{fig:numerics_vs_TW2} show an excellent agreement between numerical solutions of the full initial-boundary value problem defined in Eqs.~(\ref{non-dim}--\ref{non-dim-BC}), the numerical solutions of the travelling wave approximation provided in Eqs.~\eqref{eq:TW} and the analytical solutions of the reduced version of Eqs.~\eqref{eq:TW} obtained when one neglects $\Da,\Pe$ presented in section \ref{sec:exp} when $a=1$ and $b=1,2,3$. The parameter values for Figure \ref{fig:numerics_vs_TW1} have been chosen to be representative of those determined in relation to the experimental data in \S\ref{HgII}, while those for Figure \ref{fig:numerics_vs_TW2} have been chosen to be representative of those determined in relation to the experimental data in \S\ref{CO2}. 

The first row in both figures shows the evolution of the travelling wave along the adsorption column (blue to yellow as $\tau$ increases). We can see how the initial transient evolves until it matches the early-time solutions discussed in \S\ref{sec:early}.
%remaining in a quasi-steady state until times of order unity. 
For larger times the system follows the travelling wave form. In Figure \ref{fig:numerics_vs_TW1}, corresponding to Hg(II), the wave front becomes steeper with increasing $b$. In Figure \ref{fig:numerics_vs_TW2} the front is steepest for low $b$.

In the second row of both figures we show  curves corresponding to constant $C(\zeta^*(\tau),\tau)$. If a travelling wave occurs we should observe that
\begin{equation}
    \zeta^*(\tau) = \dot{S}\tau + \kappa\, .
\end{equation}
Taking logarithms of both sides,
\begin{equation}
    \log(\zeta^*)= \log(\dot{S}\tau+\kappa)=\log(\tau)+\log(\dot{S})+\log\left(1+\frac{\kappa}{\dot{S}\tau}\right)\, .
\end{equation}
For sufficiently large values of $\tau$ the final term is negligible and therefore, if there is a travelling wave, the three curves must become co-linear.
In the figures the curves correspond to $C(\zeta^*(\tau),\tau)=$0.1, 0.5, and 0.9.  In all cases  they coincide for $\tau=\ord{1}$, which is precisely the time regime when the travelling wave is formed after the initial transient. In each case the slope corresponds to $\dot S$. Note that although this method does not represent a rigorous proof for the existence of travelling waves, it is very useful in suggesting when they will  exist. The plots in Figures \ref{fig:numerics_vs_TW1} and~\ref{fig:numerics_vs_TW2} are consistent with the conjecture presented in \S\ref{TW} on the existence of travelling waves.

The third row shows the breakthrough curves obtained by numerically solving the system of partial differential equations~(\ref{non-dim},\ref{non-dim-BC}), the full travelling wave system~\eqref{eq:TW}, and the analytical solutions given in \S\ref{sec:exp}. The error with respect to the full numerical solution, which can be observed in the insets, is $\ll10^{-3}$, implying that there is a very good agreement between the three solutions for the two considered parameter sets. The analytical solutions can therefore be considered accurate and can be used to fit the experimental data in \S\ref{sec:results}. 

\begin{figure}[H]
    \centering
\includegraphics[width=.95\textwidth]{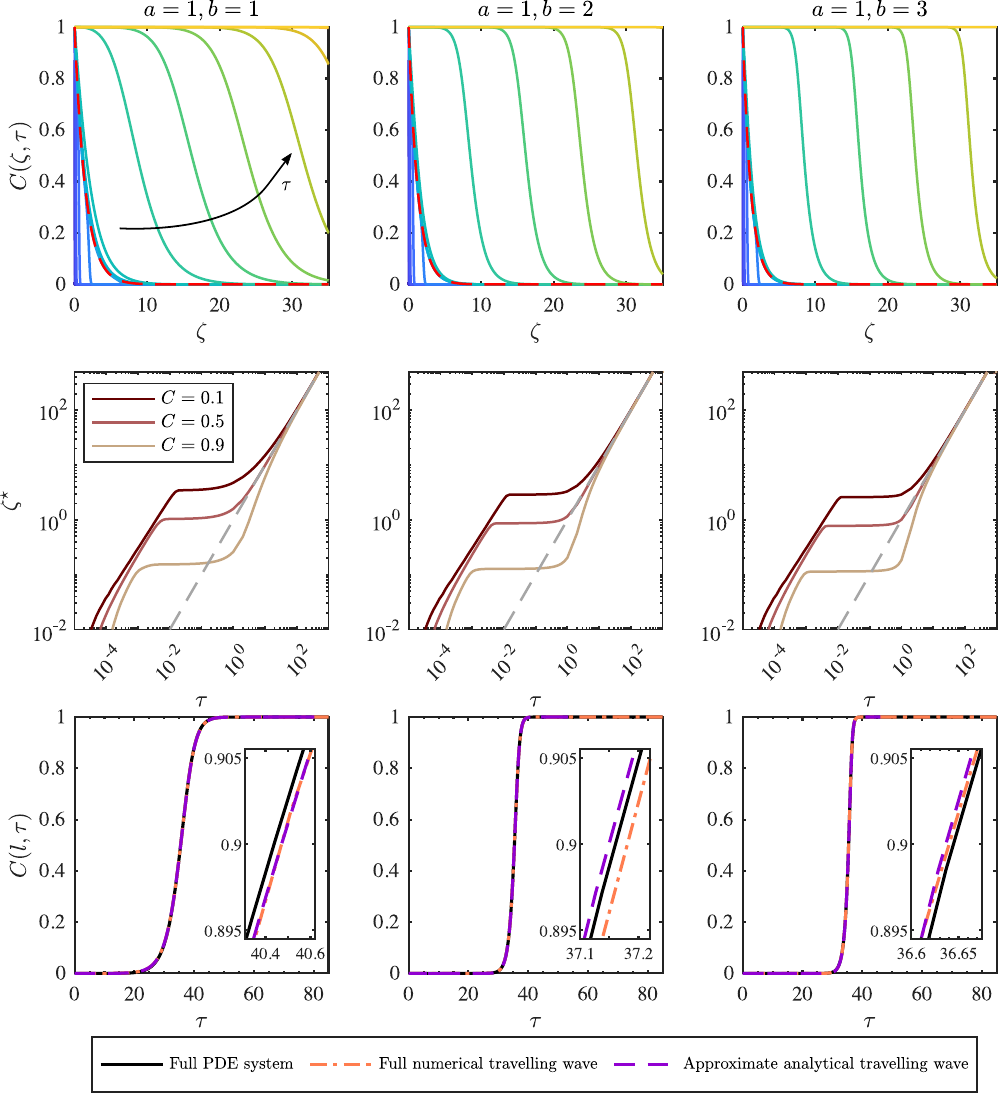} \quad
\caption{\textit{\textbf{First row}}: evolution of the travelling wave with respect to time (blue to yellow) for different combinations of $a,b$. We take logarithmically spaced times for $\Da\Pe^{-1}\leq\tau\leq1$ and linearly spaced values thereafter, $\tau\in\{2.0\times10^{-5}, 6.7 \times10^{-5}, 2.2\times10^{-4}, 7.4\times10^{-4}, 2.5\times10^{-3},8.2\times10^{-3}, 0.027, 0.09, 0.30, 1, 8.7, 16, 24, 32, 40, 47, 55\}$. The early time solution  of \S\ref{sec:early} (red dashed lines) shows the quasi-steady state of the solution profile for $\tau = O(\delta)$. \textit{\textbf{Second row}}: the three curves correspond to $C(\zeta^*(\tau),\tau)=0.1, 0.5$ and 0.9 in logarithmic scale. \textit{\textbf{Third row}}: breakthrough curves obtained by numerically solving the system of partial differential equations~(\ref{non-dim},\ref{non-dim-BC}), the full travelling wave system~\eqref{eq:TW} and the analytical solutions given in \S\ref{sec:exp}.  The parameter values are: $l=35$ (top and bottom rows) $l=500$ (middle row), $\mu=2$, $\Pe^{-1}=5\times 10^{-3}$, $\beta=1$, $\Da = 2\alpha = 4\times 10^{-3}$.}
    \label{fig:numerics_vs_TW1}
\end{figure}

\begin{figure}[H]
    \centering
\includegraphics[width=.95\textwidth]{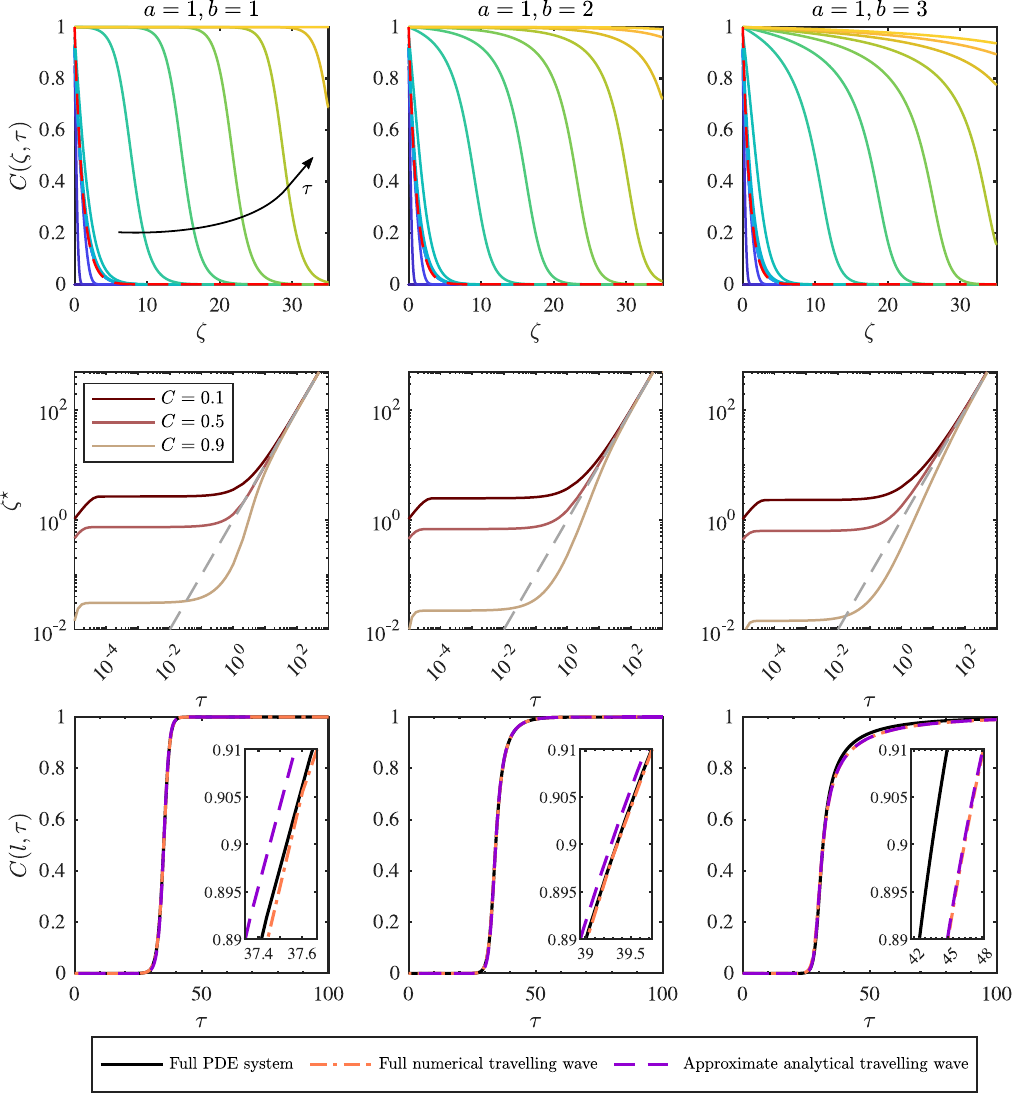} \quad
    \caption{Same as in Figure \ref{fig:numerics_vs_TW1} but with the following changes to parameter values: $\mu=2$, $\Pe^{-1}=5\times 10^{-3}$, $\beta=1$, $\Da = 2\alpha = 4\times 10^{-3}$, and, in the first row $\tau\in\{10^{-6}, 4.6\times10^{-6}, 2.2\times10^{-5}, 10^{-4}, 4.6\times10^{-4}, 2.2\times10^{-3},0.01, 0.046,0.22,1.0, 8.0, 15, 22,29,36,43,50\}$. }
    \label{fig:numerics_vs_TW2}
\end{figure}

\section{Discussion and Results} \label{sec:results}
In this section, we examine the  behaviour of the analytical models developed in \S\ref{TW} and assess their accuracy in capturing the behaviour of two different data sets. In particular, we consider data from \citet{Sulaymon2014} on the removal of the Mercury(II) ion, denoted Hg(II) (or, equivalently,  Hg$^{2+}$) by modified activated carbon, and data from  \citet{Goeppert2014} on carbon capture 
by fumed silica impregnated with polyethylenimine (PEI).  In the former case,  \citet{Sulaymon2014} state that the activated carbon adsorbent grains are approximately spherical. In the latter case, the fumed silica is essentially a powder, so the geometry of the particles is assumed to be spherical. Therefore, throughout this section we assume that the column filters have a circular cross-section and further we assume that the size and shape of each porous grain of which the column filter is comprised is spherical; this assumption means that the specific surface  $|\partial {\omega}|/|{\omega}| = 3/R$, where $R$ is now, simply, the radius of any arbitrary porous grain.

\subsection{Removal of Hg(II) on activated carbon}
\label{HgII}
\citet{Sulaymon2014} investigate the removal by adsorption of  various metallic ions from a water based solution by activated carbon.  Here, we focus on one ion in particular, Hg(II). The experimental parameters as determined in  \citet{Sulaymon2014} are summarised in Table \ref{tab:Sulaymon2014}.

\begin{table}[H]
    \centering
    \caption{Experimental parameters for the adsorption of Hg(II) by activated carbon as extracted from \cite{Sulaymon2014}.}
    \label{tab:Sulaymon2014}
    \begin{tabular}{|c|c|c|c|c|c|c|}
    \hline
        \textbf{Parameter} & \textbf{Symbol} & \textbf{SI Units} & \multicolumn{4}{c|}{\textbf{Value}} \\ \hline
        Inlet concentration & $c_{in}$ & mol/m$^3$ & \multicolumn{4}{c|}{0.2493} \\ \hline
        Diffusivity & $D$ & m/s$^2$ & \multicolumn{4}{c|}{3.055$\times10^{-9}$} \\ \hline
        Interstitial velocity & ${v}$ & m/s & \multicolumn{4}{c|}{1.186$\times10^{-3}$} \\ \hline
        Column length & ${L}$ & m & \multicolumn{4}{c|}{0.05} \\ \hline
        Cross section area & $|W|$ & m$^2$ & \multicolumn{4}{c|}{1.963$\times10^{-3}$}\\ \hline
        Bulk density & $\rho_b$ & kg/m$^3$ & \multicolumn{4}{c|}{784.34} \\ \hline
        Column porosity & $\phi$ & - & \multicolumn{4}{c|}{0.601} \\ \hline
        Particle porosity & $\phi_p$ & - & \multicolumn{4}{c|}{0.76}\\ \hline
        Particle radius ($\times10^{-4}$) & $R$ & m & 1.03 & 1.9 & 3.075 & 3.89\\ \hline
        Half time & $t_{1/2}$ & h & 1.567 & 1.230 & 0.947 & 0.680\\ \hline
        Maximum adsorption possible & 
        $\bar{m}_{max}$ & mol/kg  & 0.0482 & 0.0378 & 0.0291 & 0.0208 \\ \hline
        Adsorbed amount at equilibrium & $\bar{m}_{e}$ & mol/kg & 0.0252 & 0.0198 & 0.0152 & 0.0108\\ \hline
    \end{tabular}
\end{table}

The amount of Hg(II) adsorbed at equilibrium, $\bar{m}_{e}$, has been calculated by integrating the area over the breakthrough curve via  $$\bar{m}_{e}\equiv({v}\phi/(\rho_bL))\int^\infty_0 (c_{in}-\bar{c})\,\text{d}t.$$ Thus, the maximum adsorption is $\bar{m}_{max}\equiv\mu\bar{m}_{e}$. Using the data from the batch experiments in \citet{Yousif2013}, we take $\mu=1.9125$. Note that the authors of \citet{Yousif2013} are the same as in \citet{Sulaymon2014} and consider the same experimental conditions but for batch experiments.

Figure \ref{fig:Hgiso} shows the fitting of the Sips isotherm~\eqref{Sipsiso} for various values of $a/b$, to the experimental data extracted from \citet{Yousif2013}, obtained by using the MATLAB Curve Fitting Toolbox. Note that, in Equation~\eqref{Sipsiso} $c_e$ depends on $a$ and $b$ only through a term of the form $a/b$. Table \ref{tab:Hgisogoodfit} shows the calculated fitting parameters $\bar{m}_{max}$ and $\mathcal{K}$ and also the parameters measuring the goodness of fit. 
Note that, as expected, the predicted $\bar{m}_{max}$ values are higher for the batch experiments than they are for the column experiments.

\begin{figure}[H]
    \centering
\includegraphics[width=1\textwidth]{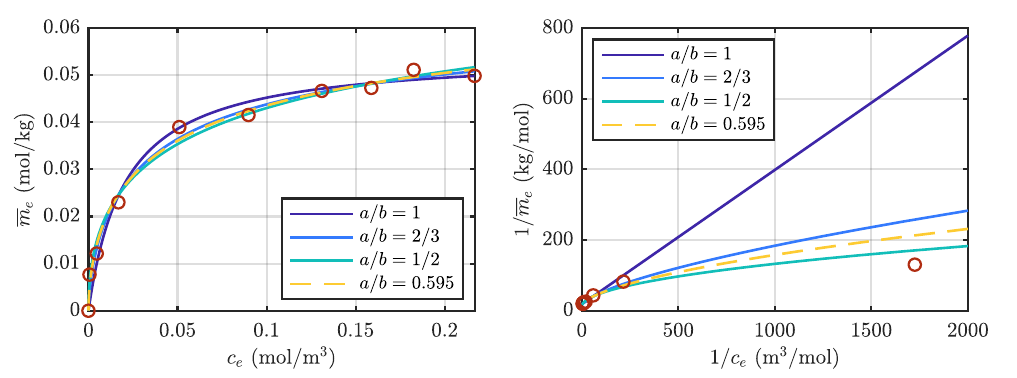} \quad
    \caption{Fitting of Sips isotherm~\eqref{Sipsiso} to the  batch data of \cite{Yousif2013} for Hg(II)  with average particle radius $R=1.875\times 10^{-4}$ m. We denote the concentration of Hg(II) at equilibrium as $c_e$. 
    }
    \label{fig:Hgiso}
\end{figure}

\begin{table}[H]
    \centering
    \caption{Sum of Squared Errors (SSE) and R-squared parameters obtained from  fitting  the Sips isotherm with various $a$ and $b$ to the equilibrium  data of  \citet{Yousif2013}.}
    \label{tab:Hgisogoodfit}
    \begin{tabular}{|c|c|c|c|c|c|}
    \hline
        \textbf{Parameter} & \textbf{Units} & \textbf{$a/b=1$} & \textbf{$a/b=1/2$} & \textbf{$a/b=2/3$} & \textbf{$a/b=0.595$} \\ \hline
        $\bar{m}_{max}$ & mol/kg & 0.0546 & 0.0902	& 0.0667 & 0.0733 \\ \hline
        $\mathcal{K}$ & m$^{3}$/mol & 48.181 & 2.883 & 8.8725 & 5.7631 \\ \hline
        SSE & (mol/kg)$^2$ & 0.580$\times10^{-4}$ & 0.309$\times10^{-4}$ & 0.294$\times10^{-4}$ & 0.281$\times10^{-4}$ \\ \hline
        R-squared & - & 0.983 & 0.991 & 0.991 & 0.992 \\ \hline
    \end{tabular}
\end{table}
By considering $a/b$ as another fitting parameter, $a/b=0.595$ yields the best fit (lowest SSE and R-squared closer to one). However, and according to the reaction~\eqref{chemistry}, this has no clear physical significance. The most physically reasonable combinations of $a$ and $b$ are $a/b=1, 1/2,$ or $2/3$; out of these options $a/b=1/2$ and $2/3$ yield an R-squared value which is closest to unity. Despite the fact that $a/b=2/3$ yields a slightly lower SSE than $a/b=1/2$, the difference between these two options is sufficiently marginal that we cannot draw a conclusion on which is the correct choice. To better distinguish between these two cases, we consider small values of $c_e$ (Figure \ref{fig:Hgiso}, right); from this we conclude that $a/b = 1/2$ offers the best fit and thus best describes the kinetics of the Hg(II) adsorption. Using $a/b=1/2$, we find that $\bar{m}_{max}=0.0902$ mol/kg. To determine 
$\mu\equiv \bar{m}_{max}/\bar{m}_{e}$ we take the value $\bar{m}_{e}=0.0472$ mol/kg as reported in \citet{Yousif2013} so that $\mu=1.9125$. Further, we take $\mu=1.9125$ to hold in the column studies. 

The reaction orders $a=1, b=2$, which correspond to the ratio $a/b=1/2$, could be related to the oxidation state of Hg(II), which according to Fourier-transform infrared spectroscopy studies carried out by  \citet{Sulaymon2014}, is complexed by H and O atoms of hydroxyl bonds in a cation exchange reaction. 

\begin{figure}[H]
    \centering
 \includegraphics[width=.95\textwidth]{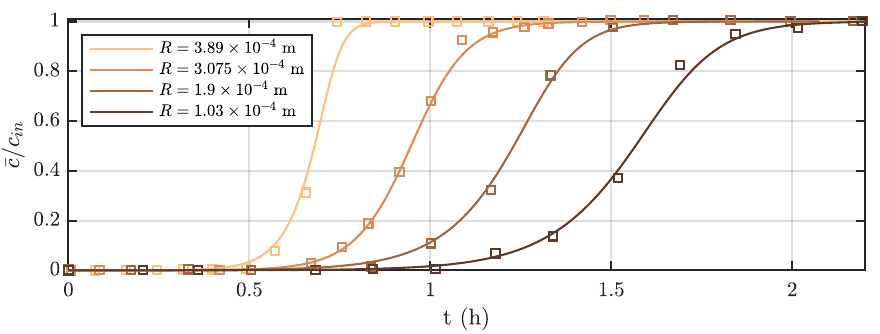}
    \caption{Fitting of the approximate analytical breakthrough model~\eqref{a1b2C} (corresponding to  $a=1$ and $b=2$) to the experimental breakthrough data for Hg(II) adsorption onto activated carbon for four distinct particle sizes \cite{Sulaymon2014}. As $R$ increases the colour transitions from dark brown to light orange.} 
    \label{fig:HgBC}
\end{figure}

On determining the kinetic orders $a=1$, $b=2$, we fit the analytical solution to the experimental breakthrough data \cite{Sulaymon2014}. The results are shown in Figure \ref{fig:HgBC}.
Each curve requires only two additional  fitting parameters 
$\mathcal{T}=  \bar{m}_{e}^{1-b}/(k_+ c_{in}^a) $ and $\beta={k}_p |\partial{\omega}|   \mathcal{T}\mathrm{Da}/( |{\omega}| \varphi)$.
The results of the fitting are shown in Table \ref{tab:HgBC}.

\begin{table}[H]
    \centering
    \caption{Parameters obtained by fitting the approximate analytical breakthrough model~\eqref{a1b2C} (corresponding to  $a=1$ and $b=2$) to the experimental breakthrough data for Hg(II) adsorption onto activated carbon for four distinct particle sizes \cite{Sulaymon2014}. }
    \label{tab:HgBC}
    \begin{tabular}{|c|c|c|c|c|c|}
    \hline
        \multirow{2}{5em}{\centering \textbf{Parameter}} & \multirow{2}{5em}{\centering \textbf{Units}} & \multicolumn{4}{c|}{\textbf{Particle Radius (m)}} \\ \cline{3-6} 
        & & \textbf{$\boldsymbol{1.03\times10^{-4}}$} & \textbf{$\boldsymbol{1.9\times10^{-4}}$} & \textbf{$\boldsymbol{3.075\times10^{-4}}$} & \textbf{$\boldsymbol{3.89\times10^{-4}}$} \\ \hline
        $\mathcal{T}$ & s & 617.37 & 364.67 & 612.38 & 83.40 \\ \hline
        $\beta$ & - & 1.6658 & 0.9834 & 4.6227 & 0.4168 \\ \hline
        $k_{+}$ & m$^3$kg mol$^{-2}$s$^{-1}$ & 0.138 & 0.233 & 0.139 & 1.020 \\ \hline
        $k_{-}$ & kg mol$^{-1}$s$^{-1}$ & 0.0286 & 0.0484 & 0.0288 & 0.2117 \\ \hline
        ${k}_p$ & m/s & 3.445$\times10^{-5}$ & 6.352$\times10^{-5}$ & 2.878$\times10^{-4}$ & 2.410$\times10^{-4}$ \\ \hline
        $\mathcal{L}$ & m & 2.97$\times10^{-3}$ & 1.75$\times10^{-3}$ & 2.94$\times10^{-3}$ & 4.01$\times10^{-4}$ \\ \hline
        Pe$^{-1}$ & - & 8.68$\times10^{-4}$ & 1.5$\times10^{-3}$ & 8.75$\times10^{-4}$ & 6.4$\times10^{-3}$ \\ \hline
        Da & - & \multicolumn{4}{c|}{4.05$\times10^{-3}$} \\ \hline
        $\alpha$ & - & \multicolumn{4}{c|}{2.04$\times10^{-3}$} \\ \hline
        SSE & - & 0.0077 & 0.0020 & 0.0043 & 0.0259 \\ \hline
        R-squared & - & 0.9966 & 0.9993 & 0.9987 & 0.9934 \\ \hline
    \end{tabular}
\end{table}

The results in Figure \ref{fig:HgBC} and 
Table \ref{tab:HgBC} show an excellent agreement between the model and the  data with $a=1$, $b=2$.
 The values of $k_{+}$ and $k_{-}$ generally increase  with increasing particle size, while ${k}_p$ decreases with increasing particle size.  The equilibrium constant $\mathcal{K}\equiv k_{+}/k_{-}$ is constant at constant temperature, so the change in $k_{-}$ is driven by that in $k_{+}$. 
Using these fitting parameters, we calculate the dimensionless values Pe$^{-1}$, Da and $\alpha$. Note that, they are all $\ord{10^{-3}} $,  and consequently  consistent with the assumptions used to derive the travelling wave equations in \S~\ref{TW}.

The value of $\beta$ is  $\ord{1}$ for all the particle sizes; this indicates that the effect of the particle size in the breakthrough curve is significant, as discussed in \S\ref{sec:param}. The value of $\beta$ also generally decreases as the particle radius ($R$) increases (see Table \ref{tab:HgBC}). The translation  of the curves  to the right as the particle size decreases shown in Figure \ref{fig:HgBC} indicates that the intra-particle diffusion is not the cause of the change in $\beta$ rather it is a result of the increase of the external specific surface when the column is filled with a greater number of small  particles. When the effect of intra-particle diffusion is small, smaller particles lead to higher adsorbed fractions. Thus, the decrease of $\beta$ in this model captures the size effect rather than the effect of intra-particle diffusion \cite{Valverde2024}.

%\editmarker{}
\subsection{Removal of $\text{CO}_\text{2}$ from air}
\label{CO2}

Goeppert et al. \cite{Goeppert2014} study the adsorption of carbon dioxide ($\text{CO}_\text{2}$)  onto fumed  silica impregnated with PEI. The addition of PEI to the  silica improves the adsorption since $\text{CO}_\text{2}$ reacts with the primary and secondary amino groups in PEI. If the reaction occurs in a dry environment, two molecules of $\text{CO}_\text{2}$ react with the amine radicals to produce carbamate molecules on the surface of the adsorbent. However in the presence of water molecules, only one amine is needed to react with one molecule of $\text{CO}_\text{2}$  \cite{Goeppert2014}. The experimental parameters as determined in  \cite{Goeppert2014} are summarised in Table \ref{tab:Goeppert2014}.

\begin{table}[t]
    \centering
    \caption{Experimental data extracted from Goeppert et al. \cite{Goeppert2014}.}
    \label{tab:Goeppert2014}
    \begin{threeparttable}
    \begin{tabular}{|c|c|c|c|c|c|c|}
    \hline
        \textbf{Parameter} & \textbf{Symbol} & \textbf{SI Units} & \multicolumn{4}{c|}{\textbf{Value}} \\ \hline
        Inlet concentration & $c_{in}$ & mol/m$^3$ & \multicolumn{4}{c|}{0.0166} \\ \hline
        Interstitial velocity & ${v}$ & m/s & \multicolumn{4}{c|}{0.2314} \\ \hline
        Column length & ${L}$ & m & \multicolumn{4}{c|}{0.1} \\ \hline
        Cross section area & $|W|$ & m$^2$ & \multicolumn{4}{c|}{5.027$\times10^{-5}$}\\ \hline
        Bulk density & $\rho_b$ & kg/m$^3$ & \multicolumn{4}{c|}{541.127} \\ \hline
        Column porosity$^*$ & $\phi$ & - & \multicolumn{4}{c|}{0.48} \\ \hline
        Particle porosity & $\phi_p$ & - & \multicolumn{4}{c|}{0.416} \\ \hline
        Diffusivity$^{**}$ ($\times10^{-4}$) & $D$ & m/s$^2$ & 0.6 & 0.9 & 2.6 & 4.1 \\ \hline   
        Particle radius ($\times10^{-4}$) & $R$ & m & 1.25 & 1.875 & 5.5 & 8.5\\ \hline
        Half time & $t_{1/2}$ & h & 10.730 & 12.188 & 12.080 & 14.028\\ \hline
        Maximum adsorption possible  &  $\bar{m}_{max}$ & mol/kg  & 1.662 & 1.725 & 1.751 & 1.915 \\ \hline
        Adsorbed amount at equilibrium & $\bar{m}_{e}$ & mol/kg & 1.614 & 1.675 & 1.7 & 1.859\\ \hline
    \end{tabular}
    \begin{tablenotes}
   \item[*] Porosity $\phi$ estimated with a column-to-particle diameter ratio $\in(4.7, 32)$ \cite{Chan2011,Mughal2012}. \item[**] Diffusivity coefficients approximated from experimental chart in \citet{Levenspiel1999}.
  \end{tablenotes}
    \end{threeparttable}
\end{table}

 \citet{Goeppert2014} provide a range of particle diameters: <0.25mm; 0.25--0.5mm; 0.5--1.7mm; and >1.7mm. To calculate the radius in Table \ref{tab:Goeppert2014}, we take the average diameters for the intermediate values and 0.25, 1.7mm for the extreme values. Although in cite{Goeppert2014}  the adsorbed amount at equilibrium is reported, experimental data accounting for the isotherm of the different particle sizes at a single temperature are not provided. Nonetheless, for particles of diameter ranging 0.25-0.50 mm,   the adsorption of $\text{CO}_\text{2}$ at equilibrium as a function of temperature is reported. Taking into account that all these measurements were made with the same inlet concentration, the maximum adsorbed fraction, as well as other thermodynamic magnitudes (such as the enthalpy of adsorption) can be obtained using the Van't Hoff equation; this reads,
\begin{align}
 \label{VantHoff1}
     \mathcal{K}=A\exp{\left(-\frac{\Delta H}{R_gT}\right)} \,,
 \end{align}
where $\mathcal{K}=k_{+}/k_{-}$ is the equilibrium constant of the Sips isotherm (m$^{3a}$mol$^{-a}$), $\Delta H$ is the enthalpy of adsorption (J mol$^{-1}$), $R_g$ is the ideal gas constant (8.314 J K$^{-1}$mol$^{-1}$), and $T$ is the temperature (K). Thus, we have
\begin{subequations}
    \begin{equation}
 \label{VantHoff2}
    \exp{\left(\frac{\Delta H}{RT}\right)}=Ac_{in}^{a}\bar{m}_{max}^{b}\left(\frac{1}{\bar{m}_{e}}-\frac{1}{\bar{m}_{max}}\right)^{b},
 \end{equation}
 which can be arranged into
     \begin{equation}
 \label{VantHoff3}
    \frac{1}{T}=\frac{R}{\Delta H}\ln{\left(Ac_{in}^{a}\bar{m}_{max}^{b}\right)}+\frac{bR}{\Delta H}\ln{\left(\frac{1}{\bar{m}_{e}}-\frac{1}{\bar{m}_{max}}\right)},
 \end{equation}
\end{subequations}
The fitting of Equation~\eqref{VantHoff3} to the experimental data provided in \cite{Goeppert2014} is shown in Figure \ref{fig:VantHoff}. This has been carried out using the MATLAB Curve Fitting Toolbox.

\begin{figure}[H]
    \centering
    \includegraphics[width=.95\textwidth]{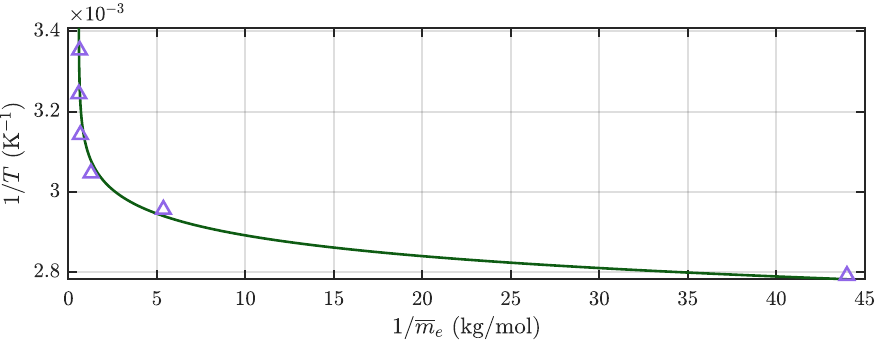}
    \caption{Temperature versus adsorbed amount  of CO$_2$ at equilibrium fitted with Equation~\eqref{VantHoff3} for particles in the range 0.25- 0.5mm.
    }
    \label{fig:VantHoff}
\end{figure}

\begin{table}[H]
    \centering
    \caption{Optimal parameters obtained via fitting of Equation~\eqref{VantHoff3} to the experimental data of \cite{Goeppert2014}.}
    \label{tab:Tvsqeparam}
    \begin{tabular}{|c|c|c|}
    \hline
        \textbf{Parameter} & \textbf{Units} & \textbf{Value} \\ \hline
        $(R/\Delta H)\ln{\left(Ac_{in}^{a}\bar{m}_{max}^{b}\right)}$ & K$^{-1}$ & 3.052$\times10^{-3}$ \\ \hline
        $bR/\Delta H$ & K$^{-1}$ & -7.15$\times10^{-5}$ \\ \hline
        $\bar{m}_{max}$ & mol/kg & 1.725 \\ \hline
    \end{tabular}
\end{table}

 Table \ref{tab:Tvsqeparam} shows the best fit parameters.
The value of $\bar{m}_{max}$ only applies to the particle size range 0.25 -- 0.5 mm; for this particle size \citet{Goeppert2014} determine $\bar{m}_{max}=1.725$ mol/kg, hence $\mu=\bar{m}_{max}/\bar{m}_{e}=1.03$. Since the maximum and the equilibrium adsorbed amount are considered to increase proportionally with increasing particle size we assume this value holds for all the particle sizes.

Note that the enthalpy of adsorption can only be obtained once the correct  $b$ is determined. In order to assess the form of the reaction, in Figure \ref{fig:abcompare} the data with $R=1.25\times10^{-4}$~m and $R=8.5\times10^{-4}$~m has been fitted by the breakthrough models~(\ref{eq:BreCu11} --\ref{eq:BreCu13C}). 

\begin{figure}[H]
    \centering
\includegraphics[width=1\textwidth]{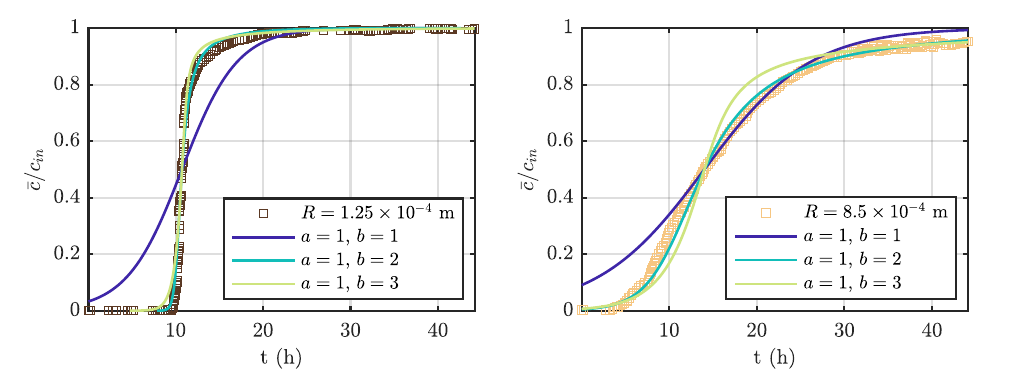} \quad
    \caption{Comparison of the breakthrough models (\ref{eq:BreCu11} --\ref{eq:BreCu13C}) to the experimental data  of  \cite{Goeppert2014}. Left: $R=1.25\times10^{-4}$ m. Right: $R=8.5\times10^{-4}$ m. Fitting parameters are presented in Table \ref{tab:abgoodfit}}
   \label{fig:abcompare}
\end{figure}

\begin{table}[H]
    \centering
    \caption{Sum of Squared Errors (SSE) and R$^2$ parameters obtained from  fitting the breakthrough models~(\ref{eq:BreCu11} --\ref{eq:BreCu13C})  to the experimental data of \cite{Goeppert2014} with $R=1.25\times10^{-4}$ m and $R=8.5\times10^{-4}$ m.}
    \label{tab:abgoodfit}
    \begin{tabular}{|c|c|c|c|c|c|}
    \hline
        \textbf{Radius (m)} & \textbf{Parameter} & \textbf{Units} & \textbf{$a=b=1$} & \textbf{$a=1, b=2$} & \textbf{$a=1, b=3$} \\ \hline
        \multirow{2}{5em}{\centering $1.25\times10^{-4}$} &
        SSE & - & 5.4559 & 0.2942 & 0.4556 \\ \cline{2-6}
        & R-squared & - & 0.7355 & 0.9857 & 0.9779 \\ \hline
        \multirow{2}{5em}{\centering $8.5\times10^{-4}$} &
        SSE & - & 0.8303 & 0.0960 & 0.6066 \\ \cline{2-6}
        & R-squared & - & 0.9617 & 0.9956 & 0.9720 \\ \hline
    \end{tabular}
\end{table}

In Table \ref{tab:abgoodfit} we present the SSE and R$^2$ values for three $a,b$ combinations and the two radii shown in Figure \ref{fig:abcompare}. 
It is clear that the best combination is $a=1,  b=2$. This is consistent with the results reported by  \citet{Aguareles2023} for small particles, which established a correlation between the reaction of $\text{CO}_\text{2}$ with the amine groups in PEI. 

\begin{figure}[H]
    \centering
\includegraphics[width=1\textwidth]{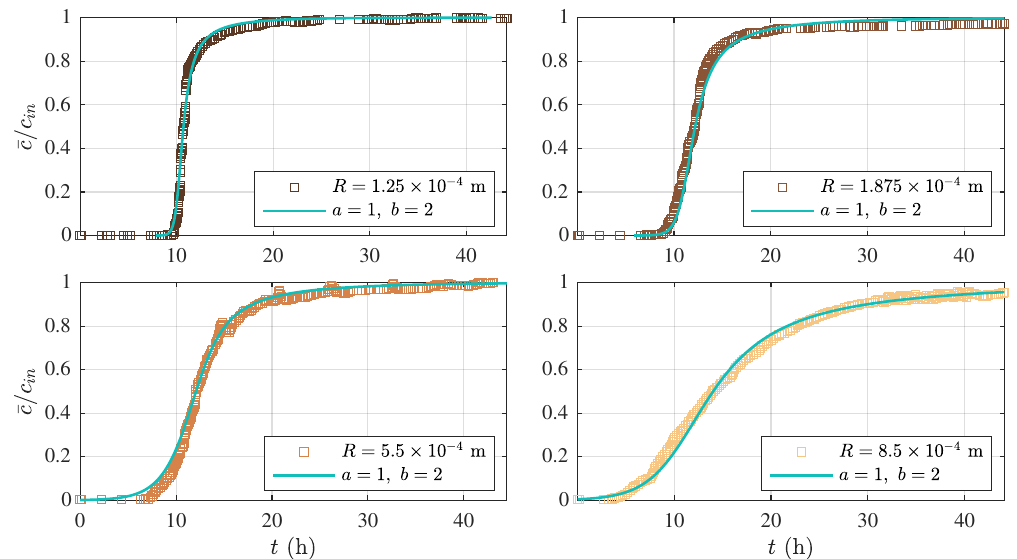} \quad
    \caption{Fitting of the breakthrough model \eqref{a1b2C} (corresponding to  $a=1$ and $b=2$) to the experimental breakthrough data of \cite{Goeppert2014} for a range of particle sizes, (a) $R=1.25\times10^{-4}$, (b) $R=1.875\times10^{-4}$, (c) $R=5.5\times10^{-4}$, (d) $R=8.5\times10^{-4}$.
    }
        \label{fig:CO2BC}
\end{figure}

The fitting of the experimental breakthrough data with the model~\eqref{a1b2C} is shown in Figure \ref{fig:CO2BC}. As before  the only fitting parameters are $\mathcal{T}$ and $\beta$; these are presented in Table \ref{tab:CO2BC}. The results shown in Figure \ref{fig:CO2BC} and  Table \ref{tab:CO2BC} demonstrate excellent agreement between the model and experimental data. The values of $k_{+}$, $k_{-}$ and ${k}_p$ clearly show an inverse relation with increasing particles size, with the exception of  $k_p$ for the largest particle. We hypothesise that this apparent  inconsistency might be related to experimental noise.

\begin{table}[H]
    \centering
    \caption{Parameters obtained from  fitting the breakthrough model~\eqref{a1b2C} to the  experimental breakthrough data of \cite{Goeppert2014}, for four particle sizes. }
    \label{tab:CO2BC}
    \begin{tabular}{|c|c|c|c|c|c|}
    \hline
        \multirow{2}{5em}{\centering \textbf{Parameter}} & \multirow{2}{5em}{\centering \textbf{Units}} & \multicolumn{4}{c|}{\textbf{Particle Radius (m)}} \\ \cline{3-6} & & \textbf{$\boldsymbol{1.25\times10^{-4}}$} & \textbf{$\boldsymbol{1.88\times10^{-4}}$} & \textbf{$\boldsymbol{5.50\times10^{-4}}$} & \textbf{$\boldsymbol{8.50\times10^{-4}}$} \\ \hline
        $\mathcal{T}$ & s & 1.07$\times10^{3}$ & 2.18$\times10^{3}$ & 2.45$\times10^{3}$ & 7.38$\times10^{3}$ \\ \hline
        $\beta$ & - & 5.93$\times10^{5}$ & 9.72 & 0.565 & 5.29 \\ \hline
        $k_{+}$ & m$^3$kg mol$^{-2}$s$^{-1}$ & 0.0348 & 0.0166 & 0.0145 & 0.00440 \\ \hline
        $k_{-}$ & kg mol$^{-1}$s$^{-1}$ & 5.13$\times10^{-7}$ & 2.45$\times10^{-7}$ & 2.14$\times10^{-7}$ & 6.49$\times10^{-8}$ \\ \hline
        ${k}_p$ & m/s & 2.33$\times10^{3}$ & 0.0294 & 0.0045 & 0.0237 \\ \hline
        $\mathcal{L}$ & m & 2.268$\times10^{-3}$ & 4.413$\times10^{-3}$ & 4.896$\times10^{-3}$ & 1.350$\times10^{-2}$ \\ \hline
        Pe$^{-1}$ & - & 0.114 & 0.0881 & 0.230 & 0.131 \\ \hline
        Da & - & 9.10$\times10^{-6}$ & 8.77$\times10^{-6}$ & 8.64$\times10^{-6}$ & 7.90$\times10^{-6}$ \\ \hline
        $\alpha$ & - & 4.10$\times10^{-6}$ & 3.95$\times10^{-6}$ & 3.90$\times10^{-6}$ & 3.56$\times10^{-6}$ \\ \hline
        SSE & - & 0.294 & 0.134 & 0.231 & 0.0960 \\ \hline
        R$^2$ & - & 0.986 & 0.994 & 0.992 & 0.996 \\ \hline
    \end{tabular}
\end{table}

In contrast to the Hg(II)  breakthrough curves in the present case  the  half-time shows little variation but the gradient decreases with increasing  particle size.  The decreasing gradient  indicates that the increase in particle size is related to a more significant intra-particle diffusion. For the smallest radius we see $\beta = \ord{10^5}$, which indicates intra-particle diffusion effects are negligible (and the model reduces to that of  \cite{Aguareles2023}). However, as the particle size increases  so that $\beta = \ord{1}$, the intra-particle model applies. The excellent data fit verifies the suitability of the model for large particles.  Finally, we note that all dimensionless values Pe$^{-1}$, Da and $\alpha$ are sufficiently small, confirming the approximations made in \S\ref{TW}.

\section{Conclusions/Future Work}\label{sec:conclusions}

This work provides a mathematical model that accounts for the effect of the micro-structure of the adsorbent in adsorption column processes combined with a Sips equation for the adsorption rate. The model has been carefully derived and  analysed and subsequently tested against experimental data. The main contributions of this work can be summarized in the following points:
 
\begin{enumerate}
    \item A rigorous derivation of the model has been presented, indicating the assumptions made at every step.  The model correctly captures the multiscale nature of the processes accounting for the advection-driven transport in the inter-particle region, the diffusion-driven transport in the intra-particle region, and the adsorption phenomenon at the adsorbent surface. 
    
    \item Adsorption kinetics have been modeled with a Sips sink term. This accounts for different possible combinations of partial orders of reaction and is able to describe both physisorption and chemisorption. The relation between the partial orders $a$ and $b$ is also linked to the adsorbent performance through the isotherm profile. 
    
    \item The existence of a solution using a travelling wave approximation has been discussed in terms of the relation between partial orders and the value of $\mu$ (the ratio of maximum to final adsorbed mass).  Analytical solutions using this approach have been provided for the most common combinations of partial-orders with "favourable adsorption" ($a\leq b$), namely $a=b=1$, $a=1$ $b=2$ and $a=1$ $b=3$. 

    \item The travelling wave solutions were compared with the numerical solution of the full mathematical model using different parameter values. The agreement between the two solution forms was excellent for all of the $a,b$ combinations studied. The numerical scheme developed in this work has been made available online \url{https://github.com/aguareles/Intra-particle-diffusion.git}. 
    
    \item Finally, the analytical solutions  have been tested against experimental data for column breakthrough curves for two different applications: Hg(II) adsorption on activated carbon from wastewater, and direct air CO$_2$ capture on PEI modified fumed silica. 
    The agreement between the breakthrough data and the analytical solution was excellent in all cases using the model with $a=1, b=2$ and only two fitting parameters. The values $a=1, b=2$ can be directly related to the reaction mechanism, namely with the ion-exchange reaction of Hg(II) that depends on the valence of the metal, and the reaction between CO$_2$ and the amine groups in the PEI structure. 
    
    The pattern observed in the fitted values with changing particle radius provides physical insights, and proves the suitability of the model to capture the effect of the decreasing specific surface and/or the increasing influence of intra-particle diffusion with increasing particle size.
\end{enumerate}

For all the reasons listed above, we can conclude that this work presents a significant contribution to the understanding of column adsorption processes and the underlying physical and chemical mechanisms. 

This work  represents a starting point for future studies where the solution under different $a,b$ configurations is assessed. This includes the study of adsorbate-adsorbent systems with unfavourable adsorption ($a>b$). In certain situations these cases may be of practical interest, for example in certain geographic locations where,  due to expense or availability, only  low quality filter materials are accessible \cite{Bishayee2022, Auton2024}. Another possible future line of research related to this work is its extension to non-integer orders, since many processes actually consist of complex mechanisms or multiple reactions which may be approximated by  global kinetics with fractional partial orders \cite{Yakabe2021}.

\section*{Declaration of Competing Interest}
The authors report no competing interests.

\section*{Acknowledgements}

This publication is part of the research projects PID2020-115023RB-I00 (funding M. Aguareles, T.G. Myers, M. Calvo-Schwarzwalder) financed by MCIN/AEI/ 10.13039/501100011033/. L. C. Auton and T.G. Myers acknowledge the CERCA Programme of the Generalitat de Catalunya for institutional support, their work was also supported by the Spanish State Research Agency, through the Severo Ochoa and Maria de Maeztu Program for Centres and Units of Excellence in R\&D (CEX2020-001084-M).
A. Valverde acknowledges support from the Margarita Salas UPC postdoctoral grants funded by the Spanish Ministry of Universities with European Union funds - NextGenerationEU. 

% This work is supported by the Spanish State Research Agency, through the Severo Ochoa and María de Maeztu Program for Centers and Units of Excellence in R&D (CEX2020-001084-M).
% We thank CERCA Programme/Generalitat de Catalunya for institutional support.

\appendix
\section{Numerical scheme} 
\label{Cheb_Num}
Here, we discuss the numerical schemes in more detail.  
\subsection{Full PDE system} 
Spectral collocation methods involve discretising the solution domain into a set of $N$ points (collocation points), defining a global function that interpolates the solution at these collocation points (the interpolant), and then approximating the derivatives of the solution as the derivatives of the interpolant. In a Chebyshev spectral collocation method, the collocation points are the $N$ Chebyshev points $x_k\in[1,-1]$, which can be defined as~\cite{piche2007solving}
\begin{equation}\label{eq:Chebpoints}
    x_k = \cos\left(\frac{(k-1)\pi}{N-1}\right), \quad k = 1,\ldots{},N.
\end{equation}
The basis functions from which the interpolant is composed are then a set of $N$ polynomials of degree $N-1$ satisfying the criterion that each is nonzero at exactly one distinct collocation point. Note that other definitions of the Chebyshev points are also commonly used (\textit{e.g.}, \citep{trefethen2000spectral}). For the definition~(\ref{eq:Chebpoints}), \citet{weideman2000matlab} provide a suite of \verb+MATLAB+ functions that generate the Chebyshev points and differentiation matrices, and that perform interpolation.

We next outline the implementation of the numerical scheme; we express  the system of  equations~\eqref{non-dim} in vector notation via
\begin{equation}\label{GenFmq}
    \bm{\mathfrak{M}}\frac{\partial\bm{Y}}{\partial{\tau}} = \bm{F}[\bm{Y}],
\end{equation}
where $\bm{Y}\defeq(C,C_p,m_p)^\intercal$ is the concatenated vector of the dependent variables,  $\bm{\mathfrak{M}}$ is a constant $3\times3$ mass matrix and where $\bm{F}$ is a  continuous vector partial-differential operator  in $\zeta$. Below, we denote spatially discretised quantities with tildes; vectors and matrices are additionally in bold. Following the method of lines, we discretise  $\zeta$,  $C$, $C_p$ and $m_p$  in space to have $N$ elements each which we denote by $\tilde{\bm{C}}$, $\tilde{\bm{C}}_p$ and $\tilde{\bm{\mathfrak{M}}}_p$. Concatenating  $\tilde{\bm{C}}$, $\tilde{\bm{C}}_p$ and $\tilde{\bm{m}}_p$  produces  a vector $\tilde{\bm{Y}}$  of length $3N$. Further we discretise $\bm{\mathfrak{M}}$ into a $3N\times3N$ constant mass matrix. Using the first and second order Chebyshev differentiation matrices of size $N \times N$, nested in matrices of size $3N\times3N$ which are otherwise zero,  we discretise the vector operator $\tilde{\bm{F}}$. 

For $\tilde{\bm{C}}_p$ and $\tilde{\bm{m}}_p$ with  $0\leq\zeta \leq l$ and for $\tilde{\bm{C}}$ with $0<\zeta<l$ we have a system of coupled ODEs in time. At $\zeta=0$ and $\zeta=l$, we enforce the spatially discretised boundary conditions on $\tilde{\bm{C}}$, 
(\textit{c.f.,} conditions (\ref{non-dim-BC}a,b)), which are algebraic in $\tau$. Thus the $3N-2$ differential equations and two algebraic equations which enforce the boundary conditions on $\tilde{\bm{C}}$, constitute a system of DAEs. 

We express this system of DAEs using a mass matrix $\tilde{\bm{\mathfrak{M}}}$, which is the Chebyshev spatial discretisation of $\bm{\mathfrak{M}}$. The mass matrix pre-multiplies the time derivative $\partial{\tilde{\bm{Y}}}/\partial{\tau}$ and enables us to enforce the boundary conditions on $\tilde{\bm{C}}$   by setting the first and $N^\text{th}$ rows of $\tilde{\bm{\mathfrak{M}}}$ identically equal to zero.
We integrate this system of DAEs in time in \verb+MATLAB+{\textsuperscript{\textregistered}} using \verb+ode15s+. 
\begin{figure}[htb]
    \centering
\includegraphics[width=1\textwidth]{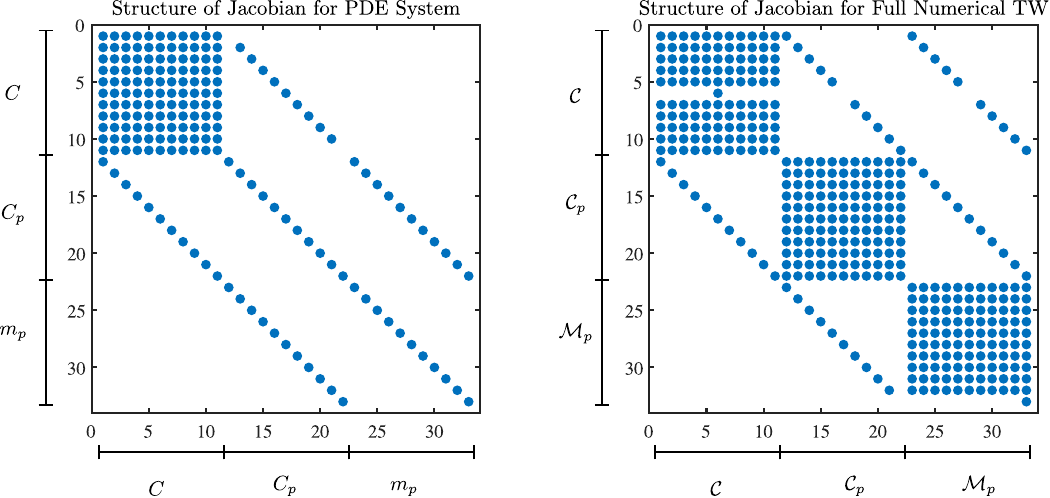} \quad
    \caption{The footprint, or sparcity pattern for the analytically determined Jacobians used in the numerical schemes,  for $N=11$.  \textit{\textbf{Left}}: Full PDE system. \textit{\textbf{Right}}: Full ODE system after travelling wave assumption.  }
        \label{Jacs}
\end{figure}

We provide the solver with an analytical Jacobian
% , $\tilde{\bm{\mathfrak{J}}}$
\begin{equation}
\label{Jac_pde}
   \tilde{\bm{\mathfrak{J}}} \defeq \frac{\mathrm{d}}{\mathrm{d} \tilde{\boldsymbol{Y}}} \left(\tilde{\boldsymbol{F}}(\tilde{\boldsymbol{Y}})\right),
\end{equation}
  The spatial discretisation has spectral accuracy,  so the overall accuracy will be determined by the accuracy of the time-stepping performed by  \verb+ode15s+.

Figure \ref{Jacs} (left) shows the footprint or  sparsity pattern of the Jacobian matrix $ \tilde{\bm{\mathfrak{J}}}$ for $N=11$ as defined in Equation~(\ref{Jac_pde}). 
  \subsection{Full system of travelling wave ODEs}
Here, we show the analogous Jacobian but for full travelling wave system (Figure \ref{Jacs}, right). 
  The structural difference between Figure \ref{Jacs} left and  right is that for the PDE we only have spatial derivatives for $C$, thus we only have the dense differentiation matrix in the top left corner, rather than for all three variables as in the ODE system. Note that the internal/boundary conditions are enforced via putting the identity element in the intersection of the row and the column that correspond to that point--- that is, for the internal condition the only entry in row $(N+1)/2$ is at column $(N+1)/2$, while for the boundary condition, the only entry in row $3N$ is at column $3N$.  All codes are available in \url{https://github.com/aguareles/Intra-particle-diffusion.git}.

\section{Proof of Proposition~\ref{prop:TW}}
\label{proofs}
We first provide two lemmas that will be used to prove Proposition~\ref{prop:TW}:
\begin{lemma}
\label{lem:poli}
    Given $\beta, \mu\in\mathbb{R}$ such that $\mu>1$ and $\beta>0$, for any $a,b\in \mathbb{N}$, and for all $x\in[0,1]$ the equation

\begin{equation}
\label{eq:pol}
\beta\left(x-y\right)  = y^a\left(\mu-x\right)^b-(\mu-1)^b x^b\, ,
    \end{equation}
    \begin{itemize}
        \item[(i)] implicitly determines a unique function $y(x)\in[0,1]$ such that $y(x)$ is continuously differentiable in the interval $[0,1]$. 
        \item[(ii)] if $a\leq b$, then $0<y(x)<x$.
    \end{itemize}
\end{lemma}
\begin{proof}
 To prove the first statement for a fixed value $x\in(0,1)$ we define 
    $$
    p_1(y) = \beta(x-y)\, ,\quad p_2(y) = y^a\left(\mu-x\right)^b-(\mu-1)^b x^b\, ,
    $$
    and  note that $p_1(y)$ is strictly decreasing while $p_2(y)$ is strictly increasing if $y\in(0,1)$. Further note that, $p_1(0)=\beta x>0$ and $p_2(0)=-x^b(\mu-1)^b<0$ if $x>0$, while $p_1(1)=\beta(x-1)<0$ and
$$p_2(1)=(\mu-x)^b-x^b(\mu-1)^b=(\mu-x)^b\left(1-x^b\left(\frac{\mu-1}{\mu-x}\right)^b\right)>0\, ,$$
if $0<x<1$. Therefore, $p_1(y),p_2(y)$ have a unique intersection point, $y\in(0,1)$. Also, we note that if $x=0$, Equation~\eqref{eq:pol} reads
$$ -\beta y  = y^a\mu^b\, ,$$
and the only solution in the interval $[0,1]$ is $y=0$. When $x=1$, Equations~\eqref{eq:pol} reads
$$ \beta\left(1-y\right)  = y^a\left(\mu-1\right)^b-(\mu-1)^b\, ,$$
whose only solution in $[0,1]$ is also $y=0$. Finally, the regularity of $y(x)$ follows from the Implicit Function Theorem.

As for the second statement, we again use the monotonicity of $p_1(y)$ and $p_2(y)$ and note that $p_1(0)>0$, $p_2(0)<0$, $p_1(x) = 0$ and 
$$p_2(x)=x^a(\mu-x)^b-x^b(\mu-1)^b=x^a(\mu-x)^b\left(1-x^{b-a}\left(\frac{\mu-1}{\mu-x}\right)^b\right)>0\, ,$$
provided $b-a\geq 0$ and $x\in(0,1)$. This shows that in this case there is an intersection point $y$ between zero and one and therefore, the solution $y(x)$ of $p(y)=0$ must be between 0 and $x$.
\end{proof}
\begin{lemma}
\label{lem:roots}
    Given $\mu\in\mathbb{R}$ such that $\mu>1$, for any $a,b\in \mathbb{N}$ such that $a\leq b$, the only roots of the polynomial
\begin{equation}
\label{eq:p}
p(x) = x^a(\mu-x)^b-(\mu-1)^b x^b\,  ,
\end{equation}
in the interval $[0,1]$ are given by $x=0$ and $x=1$. 
\end{lemma}
\begin{proof}
We note that $x=0$ is a root of $p(x)$. We then rewrite the polynomial~\eqref{eq:p} as
$$p(x) =x^a(\mu-1)^b\left[\left(\frac{\mu-x}{\mu-1}\right)^b-x^{b-a}\right]\, , $$
and define 
$$ \mathcal{P}_1(x) =\left(\frac{\mu-x}{\mu-1}\right)^b\, ,\qquad \mathcal{P}_2(x) = x^{b-a}\, .$$
If $a<b$, these polynomials satisfy, 
$$\mathcal{P}_1(1)=\mathcal{P}_2(1)=1\, \qquad\mathcal{P}_1(0)= \left(\frac{\mu}{\mu-1}\right)^b>1\, \qquad \mathcal{P}_2(0) = 0\,,$$
and therefore $x=1$ is another root of $p(x)$. The fact that $\mathcal{P}_1(x)$ is strictly decreasing and $\mathcal{P}_2(x)$ is strictly increasing between zero and one invalidates the possibility of finding any other root of $p(x)$ different than $x=0$ and $x=1$.

If $a=b$, $\mathcal{P}_2(x)\equiv 1$ and therefore, since $\mathcal{P}_1(x)$ is strictly decreasing in the interval $[0,1]$, the only intersection point between the two polynomials takes place at $x=1$.
\end{proof}

We now proof Proposition~\ref{prop:TW}, which states:

\vspace{1cm}
\textbf{Proposition 1.} Given $\mu,\beta, c_0\in\mathbb{R}$ such that $\mu>1$, $\beta>0$, and $0<c_0<1$, for any $a,b\in \mathbb{N}$ such that $a\leq b$, the initial value problem given by the equation
\begin{subequations}
    \label{eq:CCp}
    \begin{equation}
        \label{eq:CCp1}
         \frac{\dd\mathcal{C}}{\dd\eta}= -\beta\left(\mathcal{C}-\mathcal{C}_p(\mathcal{C})\right)\, ,       
    \end{equation}
along with the initial condition
\begin{equation}
    \label{eq:CCpIC}
     \mathcal{C}(0) = c_0\in(0,1)\, ,
\end{equation}
where $\mathcal{C}_p(\mathcal{C})$ is implicitly determined by
\begin{equation}
    \label{eq:CCp2}
        \beta\left(\mathcal{C}-\mathcal{C}_p\right)  = \mathcal{C}_p^a\left(\mu-\mathcal{C}\right)^b-(\mu-1)^b\mathcal{C}^b \, ,
\end{equation}
\end{subequations}
is well posed and it has a unique decreasing solution, $\mathcal{C}(\eta)$ satisfying   
\begin{equation}
\label{eq:CCpBC}
    \lim_{\eta\to -\infty} \mathcal{C}= 1\, , \quad \lim_{\eta\to \infty} \mathcal{C} = 0\, .
\end{equation}

\begin{proof}
Lemma~\ref{lem:poli} provides the existence of $\mathcal{C}_p (\mathcal{C})$ as the unique continuously differentiable solution of Equation~\eqref{eq:CCp2},
\begin{equation*}
        \beta\left(\mathcal{C}-\mathcal{C}_p\right)  = \mathcal{C}_p^a\left(\mu-\mathcal{C}\right)^b-(\mu-1)^b\mathcal{C}^b \, ,
\end{equation*}
in the interval $[0,1]$. Therefore the initial value problem~\eqref{eq:CCp},
   \begin{align*}
        & \frac{\dd\mathcal{C}}{\dd\eta}= -\beta\left(\mathcal{C}-\mathcal{C}_p(\mathcal{C})\right)\, , \\
        \label{eq:CCpIC2}
        & \mathcal{C}(0) = c_0\, ,
    \end{align*}
is well posed and it has a unique solution for any initial value $c_0\in(0,1)$. Note further that the equilibrium points of the polynomial~\eqref{eq:CCp1} satisfy $\mathcal{C}=\mathcal{C}_p$ and they are provided by the roots of the polynomial~\eqref{eq:p}, so in particular they are independent of $\beta$. This means that the only equilibrium points in the interval $[0,1]$ are $\mathcal{C}=0$ and $\mathcal{C}=1$. Lemma~\ref{lem:poli} also states that if $\mathcal{C}\in(0,1)$, then $0<\mathcal{C}_p<\mathcal{C}$ which forces the solution of the initial value problem~\eqref{eq:CCp} to be strictly decreasing. Therefore $\mathcal{C}$ connects $\mathcal{C}=1$ with $\mathcal{C}=0$ and so the boundary conditions
\begin{equation*}
    \lim_{\eta\to -\infty} \mathcal{C}= 1\, , \quad \lim_{\eta\to \infty} \mathcal{C} = 0\, ,
\end{equation*}
are satisfied.
\end{proof}

\bibliography{mainbib}

\begin{thebibliography}{34}
\providecommand{\natexlab}[1]{#1}
\providecommand{\url}[1]{\texttt{#1}}
\expandafter\ifx\csname urlstyle\endcsname\relax
  \providecommand{\doi}[1]{doi: #1}\else
  \providecommand{\doi}{doi: \begingroup \urlstyle{rm}\Url}\fi

\bibitem[Bohart and Adams(1920)]{Bohart20}
G.~S. Bohart and E.~Q. Adams.
\newblock Some aspects of the behavior of charcoal with respect to chlorine.
\newblock \emph{Journal of the American Chemical Society}, 42:\penalty0 523--544, 1920.

\bibitem[Myers(submitted to Chemical Engineering Journal, September 2023)]{Myers24}
T.~G. Myers.
\newblock Is it time to move on from the {Bohart-Adams} model for column adsorption?
\newblock Available at SSRN: \url{https://ssrn.com/abstract=4528144}, submitted to Chemical Engineering Journal, September 2023.

\bibitem[Valverde et~al.(2024)Valverde, Cabrera-Codony, Calvo-Schwarzwalder, and Myers]{Valverde2024}
A.~Valverde, A.~Cabrera-Codony, M.~Calvo-Schwarzwalder, and T.~G. Myers.
\newblock Investigating the impact of adsorbent particle size on column adsorption kinetics through a mathematical model.
\newblock \emph{International Journal of Heat and Mass Transfer}, 218:\penalty0 124724, 2024.

\bibitem[Myers et~al.(2023)Myers, Cabrera-Codony, and Valverde]{Myer22}
T.~G. Myers, A.~Cabrera-Codony, and A.~Valverde.
\newblock On the development of a consistent mathematical model for adsorption in a packed column (and why standard models fail).
\newblock \emph{International Journal of Heat and Mass Transfer}, 202:\penalty0 123660, 2023.

\bibitem[Myers and Font(2020)]{Myer20a}
T.~G. Myers and F.~Font.
\newblock Mass transfer from a fluid flowing through a porous media.
\newblock \emph{International Journal of Heat and Mass Transfer}, 163:\penalty0 120374, 2020.

\bibitem[Myers et~al.(2020)Myers, Font, and Hennessy]{Myer20b}
T.~G. Myers, F.~Font, and M.~G. Hennessy.
\newblock Mathematical modelling of carbon capture in a packed column by adsorption.
\newblock \emph{Applied Energy}, 278:\penalty0 115565, 2020.

\bibitem[Patel(2021)]{Patel2022}
H.~Patel.
\newblock Characterization and adsorptive treatment of distillery spent wash using bagasse fly ash.
\newblock \emph{Arabian Journal for Science and Engineering}, 47:\penalty0 5521--5531, 2021.

\bibitem[Langmuir(1918)]{Langmuir1918}
I.~Langmuir.
\newblock The adsorption of gases on plane surfaces of glass, mica and platinum.
\newblock \emph{Journal of the American Chemical Society}, 40\penalty0 (9):\penalty0 1361--1403, 1918.

\bibitem[Sips(1948)]{sips1948structure}
R.~Sips.
\newblock On the structure of a catalyst surface.
\newblock \emph{The Journal of Chemical Physics}, 16\penalty0 (5):\penalty0 490--495, 1948.

\bibitem[Mondal et~al.(2019)Mondal, Mondal, Kurada, Bhattacharjee, Sengupta, Mondal, Karmakar, De, and Griffiths]{Mondal19}
R.~Mondal, S.~Mondal, K.~V. Kurada, S.~Bhattacharjee, S.~Sengupta, M.~Mondal, S.~Karmakar, S.~De, and I.~M. Griffiths.
\newblock Modelling the transport and adsorption dynamics of arsenic in a soil bed filter.
\newblock \emph{Chemical Engineering Science}, 210:\penalty0 115205, 2019.

\bibitem[Seidensticker et~al.(2019)Seidensticker, Zarfl, Cirpka, and Grathwohl]{Seiden19}
S.~Seidensticker, C.~Zarfl, O.~A. Cirpka, and P.~Grathwohl.
\newblock Microplastic–contaminant interactions: {I}nfluence of nonlinearity and coupled mass transfer.
\newblock \emph{Environmental Toxicology and Chemistry}, 38\penalty0 (8):\penalty0 1635--1644, 2019.

\bibitem[Aguareles et~al.(2023)Aguareles, Barrabés, Myers, and Valverde]{Aguareles2023}
M.~Aguareles, E.~Barrabés, T.~G. Myers, and A.~Valverde.
\newblock Mathematical analysis of a {S}ips-based model for column adsorption.
\newblock \emph{Physica D: Nonlinear Phenomena}, 448:\penalty0 133690, 2023.

\bibitem[Borba et~al.(2008)Borba, {da Silva}, Fagundes-Klen, Kroumov, and Guirardello]{Borba2008}
C.~E. Borba, E.~A. {da Silva}, M.~R. Fagundes-Klen, A.~D. Kroumov, and R.~Guirardello.
\newblock Prediction of the copper ({II}) ions dynamic removal from a medium by using mathematical models with analytical solution.
\newblock \emph{Journal of Hazardous Materials}, 152\penalty0 (1):\penalty0 366--372, 2008.

\bibitem[Dantas et~al.(2011)Dantas, Luna, Silva, {de Azevedo}, Grande, Rodrigues, and Moreira]{Dantas2011}
T.~L.~P. Dantas, F.~M.~T. Luna, I.~J. Silva, D.~C.~S. {de Azevedo}, C.~A. Grande, A.~E. Rodrigues, and R.~F. P.~M. Moreira.
\newblock Carbon dioxide–nitrogen separation through adsorption on activated carbon in a fixed bed.
\newblock \emph{Chemical Engineering Journal}, 169\penalty0 (1):\penalty0 11--19, 2011.

\bibitem[Shafeeyan et~al.(2015)Shafeeyan, Daud, Shamiri, and Aghamohammadi]{Shaf15}
M.~S. Shafeeyan, W.~M. A.~W. Daud, A.~Shamiri, and N.~Aghamohammadi.
\newblock Modeling of carbon dioxide adsorption onto ammonia-modified activated carbon: Kinetic analysis and breakthrough behavior.
\newblock \emph{Energy \& Fuels}, 29\penalty0 (10):\penalty0 6565--6577, 2015.

\bibitem[Pearson(1959)]{Pear59}
J.~R.~A. Pearson.
\newblock A note on the "{Danckwerts}" boundary conditions for continuous flow reactors.
\newblock \emph{Chemical Engineering Science}, 10:\penalty0 281--284, 1959.

\bibitem[Sulaymon et~al.(2014)Sulaymon, Yousif, and {Al-Faize}]{Sulaymon2014}
A.~H. Sulaymon, S.~A. Yousif, and M.~M. {Al-Faize}.
\newblock Competitive biosorption of lead mercury chromium and arsenic ions onto activated sludge in fixed bed adsorber.
\newblock \emph{Journal of the Taiwan Institute of Chemical Engineers}, 45\penalty0 (2):\penalty0 325--337, 2014.

\bibitem[Goeppert et~al.(2014)Goeppert, Zhang, Czaun, May, Prakash, Olah, and Narayanan]{Goeppert2014}
A.~Goeppert, H.~Zhang, M.~Czaun, R.~B. May, G.~K.~S. Prakash, G.~A. Olah, and S.~R. Narayanan.
\newblock Easily regenerable solid adsorbents based on polyamines for carbon dioxide capture from the air.
\newblock \emph{ChemSusChem}, 7\penalty0 (5):\penalty0 1386--1397, 2014.

\bibitem[Yousif et~al.(2013)Yousif, Sulaymon, and Al-Faize]{Yousif2013}
S.~A. Yousif, A.~H. Sulaymon, and M.~M. Al-Faize.
\newblock Experimental and theoretical investigations of lead mercury chromium and arsenic biosorption onto dry activated sludge from wastewater.
\newblock \emph{International Review of Chemical Engineering}, 5\penalty0 (1):\penalty0 30--40, 2013.

\bibitem[Royer and Boutin(2012)]{royer2012time}
P.~Royer and C.~Boutin.
\newblock Time analysis of the three characteristic behaviours of dual-porosity media. {I}: fluid flow and solute transport.
\newblock \emph{Transport in Porous Media}, 95:\penalty0 603--626, 2012.

\bibitem[Pich{\'{e}} and Kanniainen(2007)]{piche2007solving}
R.~Pich{\'{e}} and J.~Kanniainen.
\newblock Solving financial differential equations using differentiation matrices.
\newblock In \emph{Proceedings World Congress on Engineering}, volume~II, pages 1016--1022, 2007.

\bibitem[Trefethen(2000)]{trefethen2000spectral}
L.~N. Trefethen.
\newblock \emph{Spectral Methods in {MATLAB}}.
\newblock SIAM, 2000.

\bibitem[Bj{\o}rnar{\aa} and Mathias(2013)]{bjornaraa2013pseudospectral}
T.~I. Bj{\o}rnar{\aa} and S.~A. Mathias.
\newblock A pseudospectral approach to the {McWhorter} and {Sunada} equation for two-phase flow in porous media with capillary pressure.
\newblock \emph{Computational Geosciences}, 17:\penalty0 889--897, 2013.

\bibitem[Auton and Mac{M}inn(2018)]{auton2018arteries}
L.~C. Auton and C.~W. Mac{M}inn.
\newblock From arteries to boreholes: transient response of a poroelastic cylinder to fluid injection.
\newblock \emph{Proceedings of the Royal Society A}, 474:\penalty0 20180284, 2018.

\bibitem[Auton and Mac{M}inn(2017)]{auton2017arteries}
L.~C. Auton and C.~W. Mac{M}inn.
\newblock From arteries to boreholes: steady-state response of a poroelastic cylinder to fluid injection.
\newblock \emph{Proceedings of the Royal Society A}, 473:\penalty0 20160753, 2017.

\bibitem[Pich{\'e} and Kanniainen(2009)]{piche2009matrix}
R.~Pich{\'e} and J.~Kanniainen.
\newblock Matrix-based numerical modelling of financial differential equations.
\newblock \emph{International Journal of Mathematical Modelling and Numerical Optimisation}, 1\penalty0 (1-2):\penalty0 88--100, 2009.

\bibitem[McCabe et~al.(1993)McCabe, Smith, and Harriott]{McCabe1993}
W.~L. McCabe, J.~C. Smith, and P.~Harriott.
\newblock \emph{Unit Operations of Chemical Engineering}.
\newblock McGraw-Hill, Inc., fifth edition, 1993.
\newblock ISBN 978-0-070-44844-5.

\bibitem[Chan(2011)]{Chan2011}
H.~K. Chan.
\newblock Densest columnar structures of hard spheres from sequential deposition.
\newblock \emph{Physical Review E}, 84:\penalty0 050302, 2011.

\bibitem[Mughal et~al.(2012)Mughal, Chan, Weaire, and Hutzler]{Mughal2012}
A.~Mughal, H.~K. Chan, D.~Weaire, and S.~Hutzler.
\newblock Dense packings of spheres in cylinders: Simulations.
\newblock \emph{Physical Review E}, 85:\penalty0 051305, 2012.

\bibitem[Levenspiel(1999)]{Levenspiel1999}
O.~Levenspiel.
\newblock \emph{Chemical Reaction Engineering}.
\newblock John Wiley \& Sons, Inc, third edition, 1999.

\bibitem[Bishayee et~al.(2022)Bishayee, Ruj, Nandi, Chatterjee, Mallick, Chakraborty, Nayak, and Chakrabortty]{Bishayee2022}
B.~Bishayee, B.~Ruj, S.~Nandi, R.P. Chatterjee, A.~Mallick, P.~Chakraborty, J.~Nayak, and S.~Chakrabortty.
\newblock Sorptive elimination of fluoride from contaminated groundwater in a fixed bed column: A kinetic model validation based study.
\newblock \emph{Journal of the Indian Chemical Society}, 99\penalty0 (1):\penalty0 100302, 2022.
\newblock ISSN 0019-4522.

\bibitem[Auton et~al.(2024)Auton, Mart\'{i}nez-I-\'{A}vila, Ravuru, De, Myers, and A.Valverde]{Auton2024}
L.C. Auton, M.~Mart\'{i}nez-I-\'{A}vila, S.S. Ravuru, S.~De, T.G. Myers, and A.Valverde.
\newblock Mathematical model for fluoride-removal filters.
\newblock In \emph{Progress in Industrial Mathematics at ECMI 2023}. To appear, Springer International Publishing, 2024.

\bibitem[Yakabe et~al.(2021)Yakabe, Imamura, Yoshikawa, Miyauchi, Kitajima, and Itakura]{Yakabe2021}
T.~Yakabe, G.~Imamura, G.~Yoshikawa, N.~Miyauchi, M.~Kitajima, and A.N. Itakura.
\newblock 2-step reaction kinetics for hydrogen absorption into bulk material via dissociative adsorption on the surface.
\newblock \emph{Scientific Reports}, 11\penalty0 (18836), 2021.
\newblock ISSN 2045-2322.

\bibitem[Weideman and Reddy(2000)]{weideman2000matlab}
J.~A.~C. Weideman and S.~C. Reddy.
\newblock A {MATLAB} differentiation matrix suite.
\newblock \emph{ACM Transactions on Mathematical Software (TOMS)}, 26\penalty0 (4):\penalty0 465--519, 2000.

\end{thebibliography}

\end{document}